\documentclass[runningheads,cleveref,autoref,thm-restate,a4paper]{llncs}

\usepackage{scalefnt}
\usepackage{xspace}
\usepackage{amsmath}
\usepackage{graphicx,url}
\usepackage{booktabs}
\usepackage{float}
\usepackage{algpseudocode}

\usepackage{hyperref}

\usepackage[english,algoruled,longend,ruled,vlined,linesnumbered]{algorithm2e}

\floatname{algorithm}{Listing}

\usepackage{amssymb}
\usepackage{bm}
\usepackage{dsfont}
\usepackage{hf-tikz}
\usepackage{marvosym,cite}

\definecolor{RED}{RGB}{255,0,0}

\newcommand{\satt}{3{\sf SAT}_3\xspace}
\newcommand{\NP}{\mathcal{NP}}
\newcommand{\coNP}{\mathcal{\sf coNP}}
\newcommand{\poly}{\mathcal{\sf poly}}
\newcommand{\bigO}{\mathcal{O}}

\newcommand{\bicp}{\mathcal{\sf BIC}}
\newcommand{\bimp}{\mathcal{\sf BIM}}
\newcommand{\bmp}{\mathcal{\sf BM}}
\newcommand{\cmp}{\mathcal{\sf CM}}
\newcommand{\tmp}{\mathcal{\sf TM}}
\newcommand{\rmp}{\mathcal{\sf RM}}
\newcommand{\ocmp}{\mathcal{\sf OCM}}
\newcommand{\ecd}{\mathcal{\sf ECD}}
\newcommand{\mdecd}{\mathcal{\sf MDECD}}
\newcommand{\acd}{\mathcal{\sf ACD}}

\usepackage{color}

\usepackage{tikz}
 \usetikzlibrary{calc}
 \usepackage{boxedminipage}
\usepackage{framed}
\usepackage{xspace}
\usepackage{xifthen}
 \usepackage{tabularx}

\newlength{\RoundedBoxWidth}
\newsavebox{\GrayRoundedBox}
\newenvironment{GrayBox}[1]%
   {\setlength{\RoundedBoxWidth}{.93\textwidth}
    \def\boxheading{#1}
    \begin{lrbox}{\GrayRoundedBox}
       \begin{minipage}{\RoundedBoxWidth}}%
   {   \end{minipage}
    \end{lrbox}
    \begin{center}
    \begin{tikzpicture}%
       \node(Text)[draw=black!20,fill=white,rounded corners,%
             inner sep=2ex,text width=\RoundedBoxWidth]%
             {\usebox{\GrayRoundedBox}};
        \coordinate(x) at (current bounding box.north west);
        \node [draw=white,rectangle,inner sep=3pt,anchor=north west,fill=white]
        at ($(x)+(6pt,.75em)$) {\boxheading};
    \end{tikzpicture}
    \end{center}}

\newenvironment{defproblemx}[2][]{\noindent\ignorespaces%
                                \FrameSep=6pt%
                                \parindent=0pt%
                \vspace*{-1.5em}
                \ifthenelse{\isempty{#1}}{%
                  \begin{GrayBox}{\textsc{#2}}%
                }{%
                  \begin{GrayBox}{\textsc{#2} parameterized by~{#1}}%
                }
                \newcommand\Input{Input:}%
                \begin{tabular*}{\textwidth}{@{\hspace{.1em}} >{\itshape} p{1.8cm} p{0.8\textwidth} @{}}%
            }{
                \end{tabular*}%
                \end{GrayBox}%
                \ignorespacesafterend
            }

\newcommand{\defproblemaOPT}[3]{%
  \begin{defproblemx}{#1}
    {\bf Instance:}  & #2 \\
    {\bf Goal:} & #3
  \end{defproblemx}
}%

\newtheorem{construction}{Construction}

 \hypersetup{
     colorlinks, linkcolor={dark-blue},
    citecolor={mid-green}, urlcolor={dark-blue}}

  \definecolor{mid-green}{rgb}{0.15,0.65,0.15}
 \definecolor{dark-green}{rgb}{0.15,0.25,0.15}
 \definecolor{dark-red}{rgb}{0.7,0.15,0.15}
 \definecolor{dark-blue}{rgb}{0.15,0.15,0.9}
 \definecolor{medium-blue}{rgb}{0,0,0.5}
 \definecolor{gray}{rgb}{0.5,0.5,0.5}
 \definecolor{color-Ig}{rgb}{0.15,0.7,0.15}
 \definecolor{darkmagenta}{rgb}{0.30, 0.0, 0.30}
 \definecolor{blue}{rgb}{0.15,0.15,0.9}

\renewcommand{\NP}{{\sf NP}\xspace}

\newcommand{\FPT}{{\sf FPT}\xspace}
\newcommand{\XP}{{\sf XP}\xspace}

\newcommand{\adj}{{\sf Adj}\xspace}

\begin{document}

\mainmatter  

\title{On the Complexity of the Median and Closest Permutation Problems
}

\titlerunning{\textit{On the Complexity of the Median and Closest Permutation Problems}}

\author{ 
Lu\'is Cunha\inst{1}
\and 
Ignasi Sau\inst{2}
\and
U\'everton Souza\inst{1}
}
\authorrunning{Cunha, Sau, and Souza}
\institute{
$^1$ Instituto de Computa\c{c}\~ao, Universidade Federal Fluminense, Brasil \\ 
$^2$ LIRMM, Universit\'e de Montpellier, CNRS, France\\
\email{\{lfignacio, ueverton\}@ic.uff.br}, 
\email{\{ignasi.sau\}@lirmm.fr}
}

\toctitle{Lecture Notes in Computer Science}

\maketitle

\begin{abstract}
Genome rearrangements are events where large blocks of DNA exchange places during evolution. The analysis of these events is a promising tool for understanding evolutionary genomics, providing data for phylogenetic reconstruction based on genome rearrangement measures. 
Many pairwise rearrangement distances have been proposed, based on finding the minimum number of rearrangement events to transform one genome into the other, using some predefined operation. 
When more than two genomes are considered, we have the more challenging problem of rearrangement-based phylogeny reconstruction. 
Given a set of genomes and a distance notion, there are at least two natural ways to define the ``target'' genome. On the one hand, finding a genome that minimizes the sum of the distances from this to any other, called the \emph{median genome}. On the other hand, finding a genome that minimizes the maximum distance to any other, called the \emph{closest genome}. 
Considering genomes as permutations of distinct integers, some distance metrics have been extensively studied. 
We investigate the median and closest problems on permutations over the following metrics: 
\emph{breakpoint} distance, \emph{swap} distance, \emph{block-interchange} distance, \emph{short-block-move} distance, and \emph{transposition} distance.
In biological applications some values are usually very small, such as the solution value $d$ or the number $k$ of input permutations. 
For each of these metrics and  parameters $d$ or $k$, we analyze the closest and the median problems from the viewpoint of parameterized complexity. 
We obtain the following results:
\NP-hardness for finding the median/closest permutation regarding some metrics of distance, even for only $k = 3$ permutations;
Polynomial kernels for the problems of finding the median permutation of all studied metrics, considering the target distance $d$ as parameter; 
\NP-hardness result for finding the closest permutation by short-block-moves; 
\FPT algorithms and infeasibility of polynomial kernels for finding the closest permutation for some metrics when parameterized by the target distance $d$.
\end{abstract}

\keywords{median problem, closest problem, genome rearrangements, parameterized complexity, polynomial kernel.}

\section{Introduction}
\label{sec:intro}

Ancestral reconstruction is a classic task in comparative genomics. This problem is based on consensus word analysis, with a vast applicability~\cite{cunha2013advancing,fertin2009combinatorics,pevzner2000computational}.
Genome rearrangement problems study large-scale mutations on a set of DNAs in living organisms, and have been studied extensively in computational biology and computer science fields for decades. 
From a mathematical point of view, a genome is represented by a permutation. 
Based on that, as proposed by Watterson et al.~\cite{watterson1982chromosome}, a genome rearrangement problem is interpreted as transforming one permutation into another by a minimum number of operations depending on the possible allowed rearrangements, i.e., the chosen metric. 
This model considers the following assumptions: the order of genes in each genome is known; all genomes we compare share the same gene set; all genomes contain a single copy of each gene; and all genomes consist of a single chromosome. 
So, genomes can be modeled as permutations, once each gene is encoded as an integer. 
 
Finding the minimum number of operations is equivalent to sorting the permutation with a given rearrangement. 
Many metrics received attention in recent years, and among the studied distances or sorting problems the following are the most natural ones. 
The \emph{breakpoint} distance is the number of consecutive elements in one permutation that are not consecutive in another
one. Note that on the breakpoint distance we do not apply any operation to transform a permutation into another one. 
The \emph{block-interchange} operation transforms one permutation into another one by exchanging two blocks, and generalizes the \emph{transposition} operation, where the blocks are restricted to be consecutive. 
A \emph{swap} operation is a block-interchange whose blocks are unitary. 
A \emph{short-block-move} operation is a transposition whose blocks have at most three elements. 
Concerning the computational complexity with respect to these metrics, the corresponding problems for breakpoint distance, sorting by block-interchanges, and sorting by swap can be solved in polynomial time~\cite{christie1998genome,fertin2009combinatorics}, sorting by transpositions is \NP-complete~\cite{bulteau2012sorting}, and the complexity of sorting by short-block-moves is still an open problem. 
Some restrictions or generalizations of the presented metrics have been considered and algorithmic aspects of sorting problems have been developed~\cite{labarre2020,radcliffe2005reversals}. 



If an input has more than two genomes, there are many approaches in regard to finding ancestral genomes~\cite{bader2011transposition,caprara2003reversal,cunha2020computational,cunha2019genome,haghighi2012medians,pe1998median}. 
A relevant one is the 
{\sc Median} problem, 
where, for a metric $M$, the goal is to find a solution genome that minimizes the sum of the distances between the solution and all the input genomes.

\defproblemaOPT{Metric $M$ Median}
{A set $S$ of genomes.}
{A genome that minimizes the sum of the distances, according to metric~$M$, between the solution genome and all other genomes of $S$.}

The \textsc{Breakpoint Median} problem is \NP-hard~\cite{pe1998median} for a general input. 
The \textsc{Reversal Median} and  \textsc{Transposition Median} problems are \NP-hard even when the input consists of three permutations~\cite{bader2011transposition,caprara2003reversal}. 
Prior to this work, the complexity of \textsc{Block-interchange Median} was not known.
The same applies to the \textsc{Median}
problem regarding swap or short-block-move operations.

Haghighi and Sankoff~\cite{haghighi2012medians} observed that, with respect to the breakpoint metric, a tendency for medians is to fall on or to be close to one of the input genomes, which contain no useful information for the phylogeny reconstruction. They also conjectured the same behavior concerning other metrics.
Hence, an alternative approach is to consider the {\sc Closest} problem for a fixed metric $M$,
which aims to find a genome that minimizes the maximum distance to any genome in the input, which can be seen as finding a genome
in the center of all others, i.e., a genome corresponding to the radius of the input set. 

\defproblemaOPT{Metric $M$ Closest}
{A set $S$ of genomes.} 
{A genome that minimizes the maximum distance, according to metric~$M$, between the solution genome and any other genome in $S$.}

Lanctot et al.~\cite{lanctot2003distinguishing} studied the \textsc{Closest} problem considering strings regarding the Hamming distance, and settled that this problem is \NP-hard even for binary strings. 
Popov~\cite{popov2007multiple} studied the {\sc Closest} problem considering permutations regarding the swap operation, and showed that it is \NP-hard. 
Cunha et al.~\cite{cunha2020computational} showed that the \textsc{Closest} problem is \NP-hard for several well-known genome rearrangement
distances, such as the breakpoint and the block-interchange ones. 

The parameterized complexity (see Appendix~\ref{sec:preliminariesParameterized} for the basic definitions, which can also be found in~\cite{cygan2015parameterized}) of the \textsc{Median} and  \textsc{Closest} problems has been studied mostly regarding strings on an alphabet $\Sigma$. 
These problems, and some variations, have been considered with respect to parameters that are combinations of $k$, $d$, $|\Sigma|$, and $n$, where $n$ is the length of the strings. 
Gramm et al.~\cite{gramm2003fixed,gramm2001exact} investigated the \textsc{Closest String} problem on binary strings considering some parameters, and 
showed how to solve it in linear time for fixed $d$ (the exponential growth in $d$ is bounded by $O(d^d)$), and  
when $k$ is fixed and $d$ is arbitrary. 
Fu et al.~\cite{fu2007msoar} developed a polynomial kernelization parameterized by the distance $d$ and the number of strings $k$, of size $O(k^2d\log k)$.
Basavaraju et al.~\cite{basavaraju2018kernelization} presented a comprehensive study of \textsc{Closest String} and some related problems from the kernelization complexity perspective, and showed that \textsc{Closest String} parameterized by the distance $d$ and the length of the strings $n$ does not admit a polynomial kernel under a standard complexity assumption. 

Considering the input genomes as permutations, some few results are known, such as the fact that the {\sc Swap Closest} problem is \FPT when parameterized by the number of input permutations and the solution radius~\cite{popov2007multiple}. 
On the other hand,  the {\sc Transposition Median} problem parameterized by the number of input permutations is para-\NP-hard, since Bader proved that it is \NP-hard even if the input consists of three permutations~\cite{bader2011transposition}. 
To the best of our knowledge, a multivariate investigation of the parameterized complexity of computing the median/closest genome by the considered metrics on permutations has not been thoroughly studied in the literature. 
Therefore, our goal is to map sources of computational tractability for both consensus problems ({\sc Median} and {\sc Closest}) defined above, and consequently identify features that make it tractable through the lenses of metrics over permutations and the parameterized complexity. 


\paragraph{Our contribution.} In this article we obtain the following results:

\begin{itemize}
    \item[$\bullet$] In \autoref{sec:median}, we develop polynomial kernels for finding median permutations considering swap, breakpoint, block-interchange, transposition, and short-block-move operations, all of them parameterized by the target distance $d$. This result is in sharp contrast with the fact that, as we have also managed to prove, for most of the above metrics the problem of finding the closest permutation does not admit a polynomial kernel parameterized by $d$.

    \item[$\bullet$] In \autoref{sec:median}, we prove \NP-hardness of \textsc{Block-interchange Median}, even for only $k = 3$ permutations. 
    Based on that, in \autoref{sec:closest} we are able to reduce \textsc{Block-interchange Median} to \textsc{block-interchange Closest}, as well to {\sc transposition closest}, even for only $k = 3$ permutations. 

    \item[$\bullet$] In \autoref{sec:closest}, we prove \NP-hardness of \textsc{Short-block-moves Closest}. Since it is still an open question to decide whether the sorting problem by short-block-moves can be solved in polynomial time, it is natural to consider the ``closest'' version of the problem that, somehow surprisingly, had not been considered in related previous work~\cite{cunha2020computational}. 
    
    \item[$\bullet$] In \autoref{sec:closest}, we also provide \FPT algorithms for the 
    \textsc{Closest} problem parameterized by the target distance $d$, for some of the above metrics. Our approach is inspired from \FPT algorithms for \textsc{Closest String} (see~\cite{cygan2015parameterized}).

\end{itemize}

The above results provide an accurate picture of the parameterized complexity of the considered problems with respect to the parameters $k$ and $d$. \autoref{tab:overview} summarizes the results obtained in this paper, considering only $d$ as the parameter, and the remaining open questions.

\paragraph{Organization.} Due to space limitations, some of the contents of the paper have been moved to the appendices. In \autoref{sec-prelim} we provide some preliminaries on genome rearrangements and introduce the notation used in the subsequent sections. In Appendix~\ref{sec:Operacoes} we provide a  detailed explanation on rearrangement operations, associated graphs, and bounds on the distances we deal with. In \autoref{sec:median} (resp. \autoref{sec:closest}) we present our results for the {\sc Closest} (resp. {\sc Median}) problems, and the deferred proofs from these sections can be found in Appendix~\ref{ap:median-problem} (resp. Appendix~\ref{ap:closest-problems}). For the sake of readability, the results whose proofs are in the appendices are marked with `$(\star)$'.



\begin{table}[htbp]
\centering
\begin{footnotesize}
\begin{tabular}{|c||p{2cm}|p{2cm}|p{2cm}|p{2cm}|p{2cm}|p{2cm}|}
\hline
 & Swap & Block-\-inter\-chan\-ge & Short-Block-Moves &  Transposition & Breakpoint \\
\hline


{\sc Median} &  &  &  &  &  \\

Par. d & poly kernel & poly kernel & poly kernel & poly kernel & poly kernel \\


\hline 

{\sc Closest} & \FPT & \FPT & \FPT & ??  & ?? \\

Par. d & no poly kernel & no poly kernel & no poly kernel & no poly kernel & no poly kernel \\



\hline
\end{tabular}
\end{footnotesize}
\medskip
\caption{Overview on the results obtained in this paper comparing \textsc{Median} and \textsc{Closest} problems parameterized by $d$. Open questions are represented by ??.\label{tab:overview}}
\end{table}



\section{Preliminaries on genome rearrangements}
\label{sec-prelim}

\emph{Genome rearrangements} are events where large blocks of DNA exchange places during evolution. For some genome rearrangement models, we may consider genomes as strings or permutations. 
An \emph{alphabet} $\Sigma$ is a nonempty set of letters, and a \emph{string} over $\Sigma$ is a finite sequence of letters of $\Sigma$. 
The \emph{Hamming distance of two strings} $s$ and $s'$ of the same length,  denoted by $d_H(s,s')$, is the number of mismatched positions between $s$ and~$s'$. 
The \emph{Hamming distance of a string} $s$ of length~$n$, denoted by $d_H(s)$, is the Hamming distance of $s$ and $\iota = 0^n$. 

A \emph{permutation} of length $n$ is a particular string with a unique occurrence of each letter. In other words, it is a bijection from the set $\{1, 2, \ldots, n\}$ onto itself $\pi=[\bm{\pi(0)} \,\pi(1) \, \pi(2) \, \cdots \, \pi(n) \allowbreak \bm{\pi(n+1)}]$, such that $\bm{\pi(0)}=0$ and $\bm{\pi(n+1)}=n+1$. 
The operations we consider will never act on $\bm{\pi(0)}$ nor $\bm{\pi(n+1)}$ but are important to define graphs used to determine bounds on the distances, as discussed later.
Similarly to the above, given a metric $M$ and two permutations $\pi$ and $\sigma$ of the same length, we define $d_M(\pi,\sigma)$ as their distance with respect to metric $M$, and the \emph{distance of a permutation} $\pi$ of length $m$, denoted by $d_M(\pi)$, is the distance between $\pi$ and the \emph{identity permutation} $\iota = [\bm{0}\ 1\ 2\ \cdots\ n\ \bm{n~+~1}]$.

\subsection{Sorting by rearrangement operations}
\label{sec:sorting-rearrangement}

Note that the distance between two permutations can also be called as the sorting of a permutation, once when permutations $\pi$ and $\sigma$ are given, one can relabel $\sigma$ to be equal to $\sigma\sigma^{-1} = \iota$ and then the distance between $\pi$ and $\sigma$ is the same as sorting $\pi\sigma^{-1}$, i.e., the distance between $\pi\sigma^{-1}$ and $\iota$. 
Hence, along this paper we may use interchangeably distance or sorting problem. 
Block-interchange
generalizes a transposition and generalizes also a swap operation. 
Nevertheless, with respect to the distance problem, general operations do not imply the same computational complexity of more particular operations. 
For instance, concerning the block-interchange distance, it can be computed in polynomial time~\cite{christie1998genome}, whereas computing the transposition distance is \NP-hard~\cite{bulteau2012sorting}, and computing the swap distance is a polynomial problem, as discussed later. 
On the other hand, if a distance problem is \NP-hard, then the \textsc{Median}/\textsc{Closest} problems for the same operation are also \NP-hard. 
Indeed, this follows by considering an input set of permutations consisting of two permutations $\pi, \iota$ such that $\pi\neq \iota$, and asking for a permutation with distance at most $d$ for each, for a metric $M$. 
Then, we can see that distance as the \textsc{Closest} problem with particular instances. 

\subsection{Relationship between sorting and median/closest problems}

Given a set of $k$ permutations, each of length $n$, we can store these permutations as a $k\times n$ matrix. 
 The columns of this matrix are called the \emph{columns} of the set of permutations, which are the elements in a same position over the $k$ permutations. 
For convenience, we denote by $S$ the input matrix and by $s\in S$ a permutation of the instance.


\paragraph*{Median problems.}
Caprara~\cite{caprara2003reversal} proved that the {\sc Reversal Median} problem ($\rmp$)
is \NP-complete by the following strategy. 
It begins with the {\sc Eulerian Cycle Decomposition} problem ($\ecd$), which consists in, given an Eulerian graph, find a partition of its edges into the maximum number of cycles. 
The $\ecd$ problem was proved to be \NP-complete by Holyer~\cite{holyer1981np}.
First, Caprara reduced $\ecd$ to the {\sc Alternating cycle decomposition} problem ($\acd$), which is the problem of finding a maximum cycle decomposition of a \emph{reality and desire diagram diagram} (defined in Appendix~\ref{sec:Operacoes}). 
Then, he reduced $\acd$ to the {\sc Cycle Median} problem ($\cmp$), 
which is the problem of finding a permutation that maximizes the sum of the number of cycles in the reality and desire diagram of a given set of three permutations.
Finally, Caprara reduced $\cmp$ to the {\sc Reversal Median} problem. 
In summary, he proved $\ecd \leq_p \cmp \leq_p \rmp$. 


In 2011, just before it was proved that {\sc Sorting by Transpositions} is \NP-hard~\cite{bulteau2012sorting}, Bader~\cite{bader2011transposition} proved the \NP-completeness of the {\sc Transposition Median} problem ($\tmp$ for short). 
It was based on the following definition: 
given three input permutations $\pi^1, \pi^2, \pi^3$, find a permutation $\sigma$ such that $\sum_{i=1}^{3}d_{\sf T}(\sigma, \pi^i)$ is minimized, where $d_{\sf T}(\sigma, \pi^i)$ is the transposition distance between $\sigma$ and $\pi^i$.


Bader~\cite{bader2011transposition} proved the hardness as an adaptation of Caprara's reductions considering reversals. 
This adaptation was done by reducing $\satt \leq_p \mdecd \leq_p \ocmp \leq_p \tmp$, 
where $\mdecd$ is the {\sc Marked Directed Eulerian Cycle Decomposition} 
problem, proved \NP-hard by Bader~\cite{bader2011transposition} and defined as follows. 
Let $k$ be an integer, let $G=(V,E)$ be a directed graph, and let $E_k \subseteq E$ be a subset of its edges with $|E_k|=k$. 
The edges in $E_k$ are called the marked edges of $G$. 
$(G,E_k) \in \mdecd$ if and only if $E(G)$ can be partitioned into edge-disjoint cycles such that each marked edge is in a different cycle. 
$\ocmp$ denotes the {\sc Odd Cycle Median} problem 
defined as follows. 
Let $\pi^1, \pi^2, \pi^3$ be permutations of $\{1, \cdots ,n\}$ and let $k$ be an integer. 
Then, $(\pi^1, \pi^2, \pi^3, k) \in \ocmp$ if and only if there is a permutation $\sigma$ with $\sum_{i=1}^{3} c_{\sf odd}(G(\sigma,\pi^i)) \geq k$ (because it is known that $d_{\sf T}(\pi) \geq \frac{(n+1)-c_{\sf odd}(G(\pi))}{2}$, where $c_{\sf odd}(G(\pi))$ is the number of odd cycles in the reality and desire diagram of $\pi$,  see~\autoref{Tdistance} in Appendix~\ref{sec:Operacoes}). 
Solving an $\ocmp$ instance is equivalent to finding a permutation matching $M(\sigma)$ such that $\sum_{i=1}^{3} c_{\sf odd}(G(\sigma,\pi^i))$ is maximized. 
This sum is also called the solution value of $M(\sigma)$. 
$\tmp$ was proved to be \NP-hard by a transformation from any instance $\sigma'$ that maximizes $\sum_{i=1}^{3} c_{\sf odd}(G(\sigma',\pi^i))$ to an instance that minimizes $\sum_{i=1}^{3}d_{\sf T}(\sigma', \pi^i)$. 
This could be done by ensuring that the distance between $\sigma'$ and each $\pi^i$ achieves the lower bound of~\autoref{Tdistance}. 

In order to examine $\tmp$, 
the \emph{multiple reality and desire diagram}\footnote{Also called \emph{multiple breakpoint graph} in~\cite{bader2011transposition}, but not called in this way here so as not to confuse with breakpoint distances we also deal with.} 
was used in~\cite{bader2011transposition,caprara2003reversal}. 
Given the permutations $\pi^1, \cdots, \pi^q$ each one with length~$n$, the \emph{multiple reality and desire diagram} $MG(\pi^1, \cdots, \pi^q) = (V,E)$ is a multigraph with $V = \{0, -1, +1, \cdots, -n, +n, -(n+1)\}$ and $E = M(\pi^1) \cup \cdots \cup M(\pi^q)$, i.e., the edge set is formed by the union of all permutation matchings of the permutations. 

$\mdecd$ is \NP-hard even when the degree of all nodes is bounded by four. 
Furthermore, this result still holds for graphs $G=(V,E)$ that $|V| + |E| - k$ is odd, where $k$ is the number of marked edges. 
Based on that, Bader described a polynomial transformation from $G$ being an instance of $\mdecd$ to an $MG$. 
Hence, it is necessary to guarantee conditions on graphs to be a multiple reality and desire diagram $MG$. 
To this end, we have the following properties.

\begin{lemma}[Caprara~\cite{caprara2003reversal}]\label{lm:base}
Let $V^t$ and $V^h$ be two disjoint node sets, and let $G’ = (V^t \cup V^h, M^1 \cup \cdots \cup M^q)$ be a graph, where each $M^i$ is a perfect matching, each edge in $M^i$ has color $i$, and each edge connects a node in $V^t$ with a node in $V^h$. 
Furthermore, let $H$ be a perfect matching such that each edge in $H$ connects a node in $V^t$ with a node in $V^h$, and $H \cup M^i$ defines a Hamiltonian cycle of $G'$ 
for $1 \leq i \leq q$. 
Then, there exist permutations $\pi^1, \cdots, \pi^q$ such that $G’$ is isomorphic to the $MG(\pi^1, \cdots, \pi^q)$.
\end{lemma}

Such a matching $H$ described in \autoref{lm:base} is called a \emph{base matching} of the graph. 
An important operation over $MG$ was introduced in~\cite{caprara2003reversal}.
Given a perfect matching $M$ on a node set $V$ and an edge $e = (u,v)$, the operation $M / e$ is defined as follows. 
If $e \in M$, then $M / e = M \backslash \{e\}$. 
If $e \notin M$, and $(a, u), (b, v)$ are the two edges in $M$ incident to $u$ and $v$, then $M / e = M \backslash \{(a,u), (b,v)\} \cup \{(a,b)\}$. 

\begin{lemma}[Caprara~\cite{caprara2003reversal}]
\label{lm:operataions}
Given two perfect matchings $M, L$ of a given graph $G$
and an edge $e = (u, v) \in M$  with $e \notin L$, then $M\cup L$ defines a Hamiltonian cycle of $G$ if and only if $(M / e) \cup (L / e)$ defines a Hamiltonian cycle of $G - \{u, v\}$.
\end{lemma}

Given an $MG$ graph $G = (V, M(\pi^1) \cup \cdots \cup M(\pi^q))$, the \emph{contraction} of an edge $e = (u, v)$ yields the graph $G / e = (V \backslash \{u, v\}, M(\pi^1) / e \cup \cdots \cup M(\pi^q) / e)$. 
By induction on the node size and contracting merging cycles, as a consequence of \autoref{lm:base} we have the following result.

\begin{lemma}[Bader~\cite{bader2011transposition}]\label{lm:contraction}
Let $V^t$ and $V^h$ be two disjoint sets, and let $G = (V^t \cup V^h, M^1 \cup M^2)$ be a graph where $M^1$ and $M^2$ are disjoint perfect matchings where each edge connects a node in $V^t$ with a node in $V^h$. 
If $M^1 \cup M^2$ defines an even number of even cycles on $V$, then $G$ has a base matching $H$.
\end{lemma}

\paragraph*{Closest problems.} We start with the following result. 

\begin{theorem}[Basavaraju et al.~\cite{basavaraju2018kernelization}]\label{thm:NoFPT}
{\sc Closest String} parameterized by $d$ and $n$ does not admit a polynomial kernel unless $\NP \subseteq \coNP /\poly$.
\end{theorem}


Considering {\sc Closest} problems, 
Popov~\cite{popov2007multiple} proved \NP-completeness for the {\sc Swap Closest} problem, and
Cunha et al.~\cite{cunha2020computational} developed an \NP-completeness framework of closest permutation regarding some rearrangements, such as breakpoint and block-interchange. 
The proposed reduction was the following: by considering any set of $k$ strings of length $n$, obtain a particular set of $k$ permutations of length $f(n)$, which is $f(n)=4n$ for the breakpoint case, while $f(n)=2n$ for the block-interchange case. 
Based on that transformation, 
Cunha et al.~\cite{cunha2020computational} showed a polynomial transformation where a solution for {\sc Closest string} yields a solution for {\sc Closest permutation}, as follows.

\begin{lemma}[Cunha et al.~\cite{cunha2020computational}]\label{lm:Hb}
Given a set of $k$ permutations obtained by the transformed set of binary strings, there is a breakpoint closest permutation with maximum distance equal to $2d$ if, and only if, there is a Hamming closest string with maximum distance equal to $d$.
\end{lemma}

\begin{lemma}[Cunha et al.~\cite{cunha2020computational}]\label{lm:Hbi}
Given a set of $k$ permutations obtained by the transformed set of binary strings, there is a block-interchange closest permutation with maximum distance at most $d$ if, and only if, there is a Hamming closest string with maximum distance equal to $d$.
\end{lemma}

Hence, the developed technique is a polynomial parameter transformation ({\sf PPT}) (see Appendix~\ref{sec:preliminariesParameterized})
from {\sc Closest String} to {\sc Breakpoint Closest Permutation} and to {\sc Block-interchange Closest Permutation}.

\section{Results for the {\sc Median} problems}\label{sec:median}


First, we show that {\sc Block-interchange Median} ($\bimp$ for short) 
is \NP-complete even if the input consists of only three permutations. The proof follows a similar approach considered by Caprara~\cite{caprara2003reversal} for the reversal rearrangement.

\begin{theorem}[$\star$]\label{thm:bimpNPC}
The {\sc Block-interchange Median} problem is \NP-complete even when the input consists of three permutations.
\end{theorem}

From \autoref{thm:bimpNPC}, one can conclude that $\bimp$ is para-\NP-hard when parameterized by the number of input permutations. $\tmp$ and $\bmp$ are also known to be para-\NP-hard when parameterized by the number of input permutations~\cite{bader2011transposition,pe1998median}. 

\autoref{lm:reducedBI-t}, \autoref{lm:sbmReduction}, and \autoref{lm:breakpointAdjsolution} present useful conditions to reduce the size of the input permutations in order to obtain polynomial kernels (\autoref{thm:kernel}). 


\begin{lemma}[$\star$]\label{lm:reducedBI-t}
If an adjacency occurs in all of the input permutations, then it occurs also in solutions of the {\sc Block-interchange Median} problem and the {\sc Transposition Median} problem. 
\end{lemma}

The argument of \autoref{lm:reducedBI-t} does not hold when dealing with short-block-moves, 
because a simulation operation 
(i.e., an operation that must be applied in non-reduced permutations analogous to the reduced ones) 
may be affected when it exceeds the size of a block. 
Nevertheless, an analogous result is proved in \autoref{lm:sbmReduction}. 
We define the {\sc $d$-M Median} problem as the {\sc Median} problem for a metric $M$ parameterized by the sum $d$ of the distances between the solution and all the input instances, i.e., the median solution.

\begin{lemma}\label{lm:sbmReduction}
For {\sc $d$-Short-Block-Move Median}, 
let $I$ be an interval with $6d+1$ consecutive columns where in each column of $I$ all elements are equal, and let the middle column be with element $c$, i.e., the $(3d+1)$th column of $I$ has only element $c$. 
Then there is a median solution $s^{\star}$ that satisfies following properties: 
\begin{enumerate}
    \item Element $c$ occurs in $s^{\star}$ in the same position as in the input permutations, i.e., $c$ occurs in the $(3d+1)$th column of $s^{\star}$, which is the same column of~$I$.
    \item For any element $e$ that occurs before $I$ in the input, $e$ does not occur after the $(3d+1)$th column of $I$ in~$s^{\star}$.
    \item All elements of the input that occur in $I$ and take place before (resp. after) the $(3d+1)$th column of $I$ also occur in~$s^{\star}$  
    before (resp. after) the $(3d+1)$th column. 
\end{enumerate}
\end{lemma}
\begin{proof}
For the first statement, assume that $s$ is a median solution in which the column $\ell$ that contains element $c$ in the input contains $a$ in such a position and $c$ is in a position $j$. 
Since $B$ contains $6d+1$ columns, with respect to the input permutations $a$ must be also in $I$, otherwise at least $3d/2>d$ moves would be necessary and so this is a {\sf no}-instance. 
Assume $a$ occurs in a column $p$.
Hence, for each input permutation $\pi^i$, the short-block-move distance between $s$ and $\pi^i$ must use at least $|p - \ell|/2$ operations to move $a$ to column $\ell$, plus $|\ell - j|/2$ operations to move $c$ to column~$j$. 
Based on that, we can transform $s$ into a permutation $s^{\star}$ by keeping $c$ in position $\ell$ and $a$ in position~$j$ applying the same number of short-block-moves as before, once it is necessary to apply at least $|j-p|/2$ operations. 

For the second statement, note that if there is permutation where element $e$ appears after $I$, then it is necessary to apply at least $(6d+1)/2 > d$ operations, which contradicts a median solution for {\sc $d$-Short-Block-Move Median}. 

For the third statement, assume $A$ is the interval inside $I$ before the column of $c$'s, i.e., before the $(3d+1)$th column, and $B$ is the interval inside $I$ after the column of $c$'s. 
As proved in the previous statement, 
a median solution does not have an element $e$ that occurs in the input before $I$ and in the solution it is after $I$. 
Let us consider that in $s^{\star}$ an element of $x\in A$ occurs in $B$. 
By the pigeonhole principle, one of the elements of $B$ must occur outside of $B$. Let us consider the following cases:

Case 1) An element of $B$ occurs after $B$ in $s^{\star}$. 
In this case, let us consider that $x$ moves from $A$ to $B$ and it takes place where $y\in B$ was. 
Hence, $y$ moves to another place and, also by the pigeonhole principle, some element of $B$ takes place after $B$ in $s^{\star}$. 
Since there are $3d$ elements in $B$, more than $d$ operations were necessary in total, similar to the previous second statement, which contradicts a median solution for {\sc $d$-Short-Block-Move Median} (see \autoref{fig:sbmblocks}(a)); 

Case 2) An element of $B$ occurs inside $A$ in $s^{\star}$. 
In this case, 
with respect to an element of $A$~(resp. of~$B$) of the input that occurs in $B$~(resp. in~$A$) in $s^{\star}$, by the pigeonhole principle, there must be a cycle $\mathcal{C}$ according to the moves necessary to be applied between the positions that the elements change their positions among to $s^{\star}$ and the input permutations. 
Thus, we can transform $s^{\star}$ into $s'$ by keeping all elements of $A$ (resp. of $B$) in $A$ (resp. in $B$), since in $s^{\star}$ there must be applied moves following $\mathcal{C}$ to correct the elements according to the the input columns (see \autoref{fig:sbmblocks}(b)). 
\begin{figure}[!h]
    \centering
     \vspace{-.3cm}
    \includegraphics[width=8.9cm]{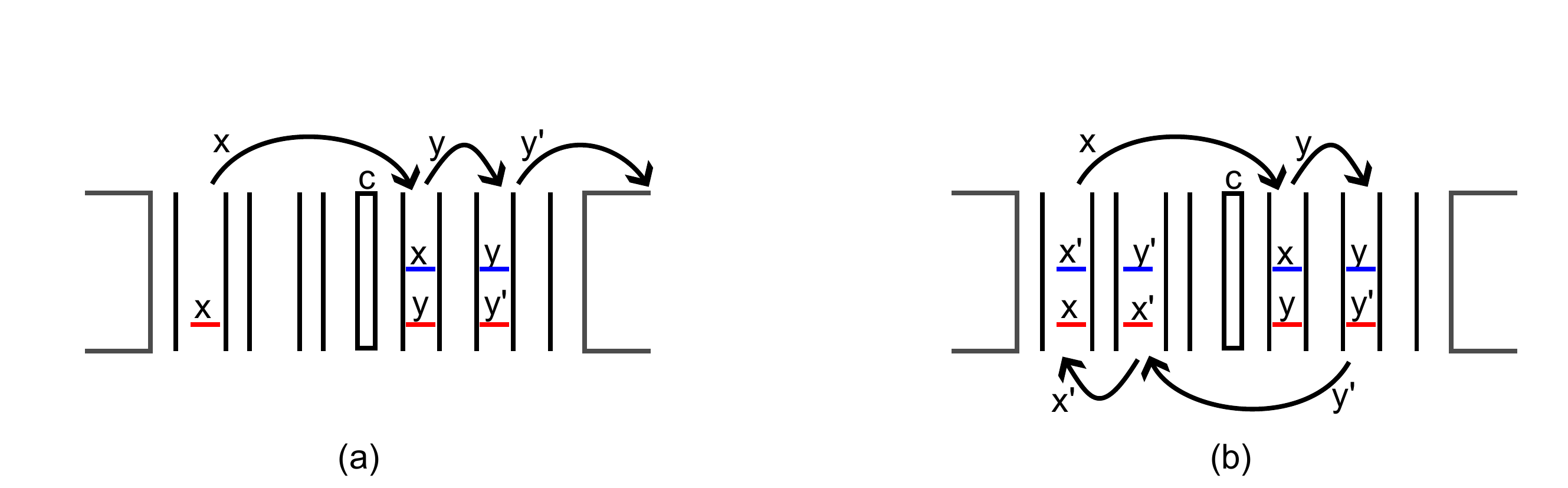}
    \caption{Cases of the third statement of \autoref{lm:sbmReduction}. In red are elements before their moves to other positions and in blue elements in $s^{\star}$ (after the moves). (a) Case 1, where element $x$ of $A$ takes place in $B$ and element $y'$ of $B$ takes place after $B$. (b) It represents whenever an element of $A$ takes place in $B$. Arrows follow a cycle $\mathcal{C}$ representing new positions of the elements. Element $x$ moves to the position where there was $y$, element $y$ moves to the position where there was $y'$, and so on. \label{fig:sbmblocks}}
    \vspace{-.4cm}
\end{figure}
~\qed\end{proof}

The complexity of the {\sc Swap Median} problem is still open even if the number of input strings is three. 
Bryant~\cite{bryant1998complexity} proved that some variations of the {\sc Breakpoint Median} problem are \NP-hard  having three input permutations, by dealing with the cases of linear, circular, signed, or unsigned permutations. 
One condition of a breakpoint median solution for given three input permutations is that if there are adjacencies common to the three input permutations, then these adjacencies can be assumed to be in a median genome~\cite{bryant1998complexity}. 
This result can be directly generalized, analogous to \autoref{lm:reducedBI-t}, as follows. 

\begin{lemma}[$\star$]\label{lm:breakpointAdjsolution}
If an adjacency occurs in all of the input permutations, then it occurs also in a solution of the {\sc Breakpoint Median} problem.
\end{lemma}

Next, we consider the parameterized complexity of some median problems parameterized by the distance $d$.
The previous lemmas allow us to develop reduction rules in order to obtain \autoref{thm:kernel}. 

\begin{theorem}[$\star$]\label{thm:kernel}
The following problems admit a polynomial kernel parameterized by the value $d$ of the desired  median solution:
{\sc $d$-Swap Median}; 
{\sc $d$-Breakpoint Median}; 
{\sc $d$-Block-interchange Median}; 
{\sc $d$-Transposition Median}; 
{\sc $d$-Short-Block-Move Median}.

\end{theorem}

\vspace{-.15cm}
\section{Results for the {\sc Closest} problems}\label{sec:closest}
\vspace{-.1cm}

First, we present a framework transformation from median to closest problem. 
Since $\bimp$ is \NP-hard even for three input permutations, we show that {\sc Block-Interchange Closest} ($\bicp$)
is \NP-hard even for three input permutations. 
This is a stronger result compared to the \NP-hardness presented by Cunha et al.~\cite{cunha2020computational} for the case where there is an arbitrary number of input permutations.

\subsection{Reducing median to closest} 

The  polynomial reduction presented in \autoref{thm:bimpNPC} allows us to show that not only the {\sc Block-Interchange Closest} problem is \NP-hard for three input permutations, but also a closest problem where the corresponding median with a constant number of input permutations is \NP-hard. 
Next we show that it is the case for the block-interchange rearrangement.

\begin{definition}
Given $\pi^1$ with $p$ elements and $\pi^2$ with $q$ elements, the \emph{union} of $\pi^1$ and $\pi^2$ is a permutation $\pi^1 \uplus \pi^2$ with $p+q+1$ elements such that $\pi^1 \uplus \pi^2 = [\pi^1_1, \pi^1_2, \cdots, \pi^1_p, (p+1), (\pi^2_1 + p+1), (\pi^2_2 + p+1), \cdots, (\pi^2_q + p+1)]$. For simplicity, $\pi^1 \uplus \pi^2$ is denoted by $\pi^{1,2}$. Permutations $\pi^1$ and $\pi^2$ are called \emph{parts of the union}.
\end{definition}

\begin{lemma}[$\star$]\label{lm:sumequal}
Given permutations $\pi^1$ and $\pi^2$, we have that $d_{\sf BI}(\pi^{1,2}) = d_{\sf BI}(\pi^1) + d_{\sf BI}(\pi^2)$.
\end{lemma}

\begin{theorem}[$\star$]\label{thm:closestmedian}
For any three permutations $\pi^1, \pi^2, \pi^3$, we have that $\sigma$ is a solution of $\bimp$ if and only if $\biguplus\limits_{i=1}^{6} \sigma$ is a solution of $\bicp$ for the permutations 
$\pi^{1,2,3,1,2,3}, \pi^{2,1,1,3,3,2}$, and $\pi^{3,3,2,2,1,1}$.
\end{theorem}

Since {\sc Block-Interchange Median} is \NP-complete (\autoref{thm:bimpNPC}), as aconsequence of \autoref{thm:closestmedian}, we have \autoref{corollary:block-interchange}.


\begin{corollary}\label{corollary:block-interchange}
The {\sc Block-Interchange Closest} problem is \NP-hard even when the input consists of three permutations.
\end{corollary}

When it is asked about other rearrangements, dealing with transpositions, sorting each part of a union separately does not yield  an optimum sequence in order to sort a permutation in general, as proved by Cunha et al.~\cite{cunha2013advancing}. Hence, an analogous strategy of the one in \autoref{thm:closestmedian} does not apply to reduce the median to the closest problems regarding transpositions rearrangement, given that \autoref{lm:sumequal} does not hold for sorting by transpositions. 
However, if each part of a union is a \textit{hurdle-free}
permutation, i.e., a permutation in which the transposition distance is equal to the lower bound on the transposition distance $d_{\sf T}(\pi) \geq \frac{(n+1)-c_{\sf odd}(G(\pi))}{2}$), then it follows that $d_{\sf BI}(\pi^{1,2}) = d_{\sf BI}(\pi^1) + d_{\sf BI}(\pi^2)$ in the same matter as \autoref{thm:closestmedian}. 
Therefore, we have \autoref{cor:TranspClosest3NPc}.

\begin{corollary}[$\star$]\label{cor:TranspClosest3NPc}
{\sc Transposition Closest} is \NP-hard even when the input consists of three permutations which are unions of hurdle-free permutations.
\end{corollary}

\subsection{The {\sc Short-block-move Closest} problem}

\paragraph*{Sufficient condition to sort by short-block-moves.} 
We refer to block-moves that introduce elements in connected components of the permutation as \emph{merging moves}. For instance, $[2 \ 3 \ \underline{1 \ 6} \ \underline{4} \ 5]\rightarrow [2 \ 3 \ 4 \ 1 \ 6 \ 5]$ is a merging move.

\begin{lemma}[$\star$]\label{thm:sorting-each-cc-separately-is-optimal}
 For every permutation $\pi$, sorting each connected component of $\pi$ separately is optimal.
\end{lemma}

There exist cases where allowing merging moves still yields an optimal solution. This is the case for $[2 \ 1 \ 4 \ 3]$, which can be sorted optimally as follows: $[2 \ \underline{1 \ 4} \ \underline{3}] \rightarrow [\underline{2 \ 3} \ \underline{1} \ 4]\rightarrow~\iota$, where $\iota = [1\ 2\ \cdots \ n]$. 
It is natural to wonder whether \autoref{thm:sorting-each-cc-separately-is-optimal} generalizes to $p$-bounded block-move, for $p>3$. However, the following counterexample shows that it is not the case, even when a block-move is bounded 
by four (i.e. a $4$-bounded block move): sorting each component of $[3 \ 2 \ 1 \ 6 \ 5 \ 4]$ separately yields a sequence of length four, but one can do better by merging components as follows: $[3 \ 2 \ \underline{1 \ 6} \ \underline{5 \ 4}] \rightarrow [3 \ \underline{2 \ 5} \ \underline{4 \ 1} \ 6] \rightarrow [\underline{3 \ 4} \ \underline{1 \ 2} \ 5 \ 6] \rightarrow \iota$.

\paragraph*{The {\sc Short-block-move Closest Permutation} problem is \NP-hard.}

First, we apply \autoref{alg:sBI} to transform any string $s$ of length $m$ into a particular permutation $\lambda_s$ of length $2m$.



\begin{algorithm}[!ht]
\SetKwInOut{Input}{Input} \SetKwInOut{Output}{Output}
\Input{A binary string $s$ of length $m$.}\Output{A permutation $\lambda_s$.}

For each occurrence of $0$ in position $i$ of $s$, set the elements $2i-1$ and $2i$ in positions $2i-1$ and $2i$ of $\lambda_s$, respectively.

For each occurrence of $1$ in position $i$ of $s$, set the elements $2i-1$ and $2i$ in positions $2i$ and $2i-1$ of $\lambda_s$, respectively.
\caption{${\sf Permut}_{\sf BI}(s)$}
\label{alg:sBI}
\end{algorithm}



Since from \autoref{thm:sorting-each-cc-separately-is-optimal} each connected component can be sorted separately, and each bit set to $1$ in $s$ corresponds to an inversion in $\lambda_s$ from \autoref{alg:sBI}, 
then it implies \autoref{lm:lengthreduced}, which is an equality between the Hamming distance of an input string $s$ and the short-block-move distance of its output permutation~$\lambda_s$.



\begin{lemma}\label{lm:lengthreduced}
 Given a string 
 of length $m$ and a permutation $\lambda_s$ of length $2m$ obtained by \autoref{alg:sBI}, the short-block-move distance of $\lambda_s$ is $d_{{\sf sbm}}(\lambda_s) = d_H(s)$.
\end{lemma}


\begin{lemma}[$\star$]\label{lm:iffBI}
Given a set of $k$ permutations obtained by \autoref{alg:sBI}, there is a short-block-move closest permutation with maximum distance at most $d$ if and only if there is a Hamming closest string with maximum distance equal to~$d$.
\end{lemma}

From \autoref{lm:iffBI} and since {\sc Hamming Closest String} is \NP-com\-plete~\cite{lanctot2003distinguishing}, we have \autoref{thm:sbmCl}. 

\begin{theorem}\label{thm:sbmCl}
{\sc Short-block-move Closest Permutation} is \NP-hard.
\end{theorem}


\autoref{thm:NoFPT}, proved by Basavaraju et al.~\cite{basavaraju2018kernelization}, states that {\sc Closest String} does not admit a polynomial kernel, unless $\NP \subseteq \coNP / \poly$. 
Since the results presented in \autoref{lm:Hb}, \autoref{lm:Hbi}, \autoref{lm:iffBI}, \autoref{thm:sbmCl}, \autoref{cor:TranspClosest3NPc}, as well as the results from Popov~\cite{popov2007multiple}, 
are {\sf PPT} reductions from {\sc Closest String}, we have the following corollary.

\begin{corollary}
{\sc Breakpoint Closest}, {\sc Block-interchange Closest}, {\sc Transposition Closest}, {\sc Swap Closest} and {\sc Short-block-move Closest} do not admit polynomial kernel unless $\NP \subseteq \coNP/\poly$.
\end{corollary}

\subsection{\FPT algorithms}

Popov~\cite{popov2007multiple} solved the {\sc Swap Closest} problem in time $O(kn + g(k,d))$ parameterized by the number of permutations $k$ (each of them of size $n$) and the distance $d$, where $g$ is a function which depends only on $k$ and $d$. 
Now, we propose \FPT algorithms for finding closest permutations of a given set of permutations, parameterized just by the distance $d$ (recall the usage of $\bigO^*(f(n))$
to say that there exists an algorithm which runs in time
$\bigO(f(n))\cdot poly(n)$, where $poly(n)$ is a polynomial function in $n$). 
Our approach is inspired by the algorithm for the {\sc Closest String}  problem~\cite{cygan2015parameterized,gramm2003fixed}, considering the three metrics below. 


\begin{theorem}\label{thm:fptAlgs}
{\sc $d$-Swap Closest}, 
     {\sc $d$-Short-Block-Move Closest}, and 
    \item {\sc $d$-Block-interchange Closest} can be solved in time  $\bigO^*(d)^{\bigO(d)}$ when parameterized by the distance $d$. 
\end{theorem}
\begin{proof}
First, we consider the {\sc Swap Closest} problem. The other problems follow in a similar way, as we discuss below. Let $\pi^1, \cdots, \pi^k$ be the input permutations. 
Recursively, we solve these problems using a bounded search tree technique as follows: 
First, set $z = \pi^1$ as a candidate permutation solution and $\ell = d$.
If $d_{\sf swap}(y, z) \leq \ell$ for each permutation $y$ of the input, then return {\tt yes}. 
Otherwise, if $\ell = 0$ then return {\tt no}. 
In the remaining case, $\ell > 0$ and there exists a permutation $\pi^i$ with $d_{\sf swap}(\pi^i, z) > d$. 
From the triangular inequality, $d_{\sf swap}(\pi^i, z) \leq 2d$ for each input permutation $z$; otherwise, the answer is {\tt no}.
Since each swap operation corrects at most two positions,  there are at most $4d$ positions on which $\pi^i$ and $z$ differ.
Let $P$ be a set of $4d$ positions on which $\pi^i$ and $z$ differ.
Hence, we branch into $|P| = 4d$ subcases: for every $p \in P$, we define $z^p$ to be equal to $z$ except for the swap putting the element $\pi_p^i$ in the position $p$ of $z^p$, and we recursively solve the problem for the pair $(z^p, \ell -1)$. 

We build a search tree of depth at most $d$, and every node has at most $4d$ children. Thus, the size of the search tree does not exceed $O((4d)^d)$. 

For the {\sc Short-Block-Move Closest} problem, it is known that each operation involves at most two edges on the associated permutation graph. Since the current solution must be at distance at most $2d$ from any input permutation, it holds that the associated permutation graph between a current solution $z$ and any $\pi^i$ has at most $4d$ edges and at most $8d$ vertices incident to some edge of the associated permutation graph; otherwise the answer is {\tt no}. Therefore, either we are already dealing with an instance with universes of small size, or there are many isolated vertices in the associated permutation graph. By the definition of permutation graphs of strings, these isolated vertices represent positions that coincide in both permutations, and we may assume that they are not involved in any move to obtain one from the other. Since we do not need to consider moves involving these isolated vertices of the associated permutation graph, we can consider only moves involving $O(d)$ many vertices. Thus, we can perform a similar bounded search tree algorithm as previously described. 


For the {\sc Block-Interchange Closest} problem, it is known that each operation changes the number of cycles in the reality and desire diagram 
by $-2$,$ 0$, or $+2$ (see~\cite{christie1998genome}). 
Moreover, from \autoref{thm:BIdistance}, there exists an optimum sequence of block-interchanges that only applies $2$-moves, i.e., each operation increases the number of cycles by two. 
This implies that there is no optimum sequence that uses $-2$ or $0$ moves. 
Recall that we obtain a sorted permutation when we achieve only cycles of size one ($n+1$ cycles in total); so, sorting is equivalent to maximizing the number of cycles in the reality and desire diagram. 
Thus, our focus is only analyzing possible $2$-moves to approximate one permutation to another one in our bounded search tree algorithm. 
It is known that there is no $2$-move that affects a $1$-cycle (cycle of length one in the diagram), 
because a $2$-move can only be performed into a unique cycle (cf.~\cite{christie1998genome,cunha2013advancing}). Thus, there is no $2$-move that affects an adjacency (a pair is an adjacency if and only if it yields a $1$-cycle in the reality and desire diagram~\cite{Bafna1998Sorting,christie1998genome}). 

At this point, we have that we can safely reduce the permutation, since all optimum block-interchange sequences do not affect adjacencies (this is a stronger result than \autoref{BIequalityreduced}). 
Hence, as each block-interchange affects at most four breakpoints, the permutation must have at most $8d$ breakpoints (i.e, $8d+1$ elements in the reduced permutation). 
Therefore, we can consider only moves involving $O(d)$ many breakpoints. Thus, we can perform a similar bounded search tree algorithm as previously described. 
~\qed\end{proof}

Note that {\sc Breakpoint Closest} does not admit a bounded search tree analogous to the ones used in \autoref{thm:fptAlgs}, since this metric does not have a sequence of operations to transform a permutation into another one; so, it is unclear how to branch. Also, for the {\sc Transposition Closest} problem, it is known that there may exist optimum sequences of transpositions that apply $0$-moves and $2$-moves, and it is a long time open question if there are optimum sequences using $-2$-moves~\cite{Bafna1998Sorting,cunha2013advancing}; so, it seems that is not safe to use the reduced permutation in that case, because there may exist an optimum sequence of transpositions that uses moves not preserved in the reduced instance, and those moves could be good for our branch step. Therefore, we leave both cases as open questions.






\bibliography{latin}

\begin{thebibliography}{10}

\bibitem{bader2011transposition}
Martin Bader.
\newblock The transposition median problem is {NP}-complete.
\newblock {\em Theor. Comput. Sci.}, 412(12-14):1099--1110, 2011.

\bibitem{Bafna1998Sorting}
Vineet Bafna and Pavel~A Pevzner.
\newblock Sorting by transpositions.
\newblock {\em SIAM J. Discrete Math.}, 11(2):224--240, 1998.

\bibitem{basavaraju2018kernelization}
Manu Basavaraju, Fahad Panolan, Ashutosh Rai, MS~Ramanujan, and Saket Saurabh.
\newblock On the kernelization complexity of string problems.
\newblock {\em Theor. Comput. Sci.}, 730:21--31, 2018.

\bibitem{bodlaender2011kernel}
Hans~L. Bodlaender, St{\'e}phan Thomass{\'e}, and Anders Yeo.
\newblock Kernel bounds for disjoint cycles and disjoint paths.
\newblock {\em Theor. Comput. Sci.}, 412(35):4570--4578, 2011.

\bibitem{bryant1998complexity}
David Bryant.
\newblock The complexity of the breakpoint median problem.
\newblock {\em Centre de recherches mathematiques, Technical Repert}, 1998.

\bibitem{bulteau2012sorting}
Laurent Bulteau, Guillaume Fertin, and Irena Rusu.
\newblock Sorting by transpositions is difficult.
\newblock {\em SIAM J. Discrete Math.}, 26(3):1148--1180, 2012.

\bibitem{caprara2003reversal}
Alberto Caprara.
\newblock The reversal median problem.
\newblock {\em INFORMS J. Comput.}, 15(1):93--113, 2003.

\bibitem{christie1998genome}
David~Alan Christie.
\newblock {\em {G}enome {R}earrangement {P}roblems}.
\newblock University of Glasgow (United Kingdom), 1998.

\bibitem{cunha2020computational}
Lu{\'\i}s Felipe~I Cunha, Pedro Feij{\~a}o, Vin{\'\i}cius~F dos Santos, Luis Antonio~B Kowada, and Celina~MH de~Figueiredo.
\newblock On the computational complexity of closest genome problems.
\newblock {\em Discrete Applied Mathematics}, 274:26--34, 2020.

\bibitem{cunha2013advancing}
Lu\'is Felipe~I Cunha, Luis Antonio~B Kowada, Rodrigo de~A.~Hausen, and Celina~MH de~Figueiredo.
\newblock Advancing the transposition distance and diameter through lonely permutations.
\newblock {\em SIAM J. Discrete Math.}, 27(4):1682--1709, 2013.

\bibitem{cunha2014faster}
Lu{\'\i}s Felipe~I Cunha, Luis Antonio~B Kowada, Rodrigo de~A.~Hausen, and Celina~MH de~Figueiredo.
\newblock A faster 1.375-approximation algorithm for sorting by transpositions.
\newblock In {\em WABI 2014}, pages 26--37. Springer Berlin Heidelberg, 2014.

\bibitem{cunha2015faster}
Lu{\'\i}s Felipe~I Cunha, Luis Antonio~B Kowada, Rodrigo de~A Hausen, and Celina~MH De~Figueiredo.
\newblock A faster 1.375-approximation algorithm for sorting by transpositions.
\newblock {\em J. Comput. Biol.}, 22(11):1044--1056, 2015.

\bibitem{cunha2019genome}
Lu{\'\i}s Felipe~I Cunha and F{\'a}bio Protti.
\newblock Genome rearrangements on multigenomic models: Applications of graph convexity problems.
\newblock {\em J. Comput. Biol.}, 26(11):1214--1222, 2019.

\bibitem{cygan2015parameterized}
Marek Cygan, Fedor~V Fomin, {\L}ukasz Kowalik, Daniel Lokshtanov, D{\'a}niel Marx, Marcin Pilipczuk, Micha{\l} Pilipczuk, and Saket Saurabh.
\newblock {\em {P}arameterized {A}lgorithms}, volume~5.
\newblock Springer, 2015.

\bibitem{downey2012parameterized}
Rodney~G Downey and Michael~Ralph Fellows.
\newblock {\em {P}arameterized {C}omplexity}.
\newblock Springer Science \& Business Media, 2012.

\bibitem{fertin2009combinatorics}
Guillaume Fertin, Anthony Labarre, Irena Rusu, St{\'e}phane Vialette, and Eric Tannier.
\newblock {\em Combinatorics of {G}enome {R}earrangements}.
\newblock MIT press, 2009.

\bibitem{fu2007msoar}
Zheng Fu, Xin Chen, Vladimir Vacic, Peng Nan, Yang Zhong, and Tao Jiang.
\newblock Msoar: a high-throughput ortholog assignment system based on genome rearrangement.
\newblock {\em J. Comput. Biol.}, 14(9):1160--1175, 2007.

\bibitem{gramm2003fixed}
Gramm, Niedermeier, and Rossmanith.
\newblock Fixed-parameter algorithms for closest string and related problems.
\newblock {\em Algorithmica}, 37:25--42, 2003.

\bibitem{gramm2001exact}
Jens Gramm, Rolf Niedermeier, and Peter Rossmanith.
\newblock Exact solutions for closest string and related problems.
\newblock In {\em ISAAC 2001}, pages 441--453. Springer, 2001.

\bibitem{haghighi2012medians}
Maryam Haghighi and David Sankoff.
\newblock Medians seek the corners, and other conjectures.
\newblock In {\em BMC bioinformatics}, volume~13, pages 1--7. Springer, 2012.

\bibitem{Vergara2}
Lenwood~S Heath and John Paul~C Vergara.
\newblock Sorting by bounded block-moves.
\newblock {\em Discrete Appl. Math.}, 88:181--206, 1998.

\bibitem{Vergara}
Lenwood~S Heath and John Paul~C Vergara.
\newblock Sorting by short block-moves.
\newblock {\em Algorithmica}, 28:323--352, 2000.

\bibitem{holyer1981np}
Ian Holyer.
\newblock The {NP}-completeness of some edge-partition problems.
\newblock {\em SIAM Journal on Computing}, 10(4):713--717, 1981.

\bibitem{book:Knuth}
D~Knuth.
\newblock {T}he {A}rt of {C}omputer {P}rogramming: {S}orting and {S}earching, vol 3, 1998.

\bibitem{labarre2020}
Anthony Labarre.
\newblock {Sorting by Prefix Block-Interchanges}.
\newblock In Yixin Cao, Siu-Wing Cheng, and Minming Li, editors, {\em 31st International Symposium on Algorithms and Computation ISAAC 2020}, volume 181, pages 55:1--55:15, 2020.

\bibitem{lanctot2003distinguishing}
J~Kevin Lanctot, Ming Li, Bin Ma, Shaojiu Wang, and Louxin Zhang.
\newblock Distinguishing string selection problems.
\newblock {\em Inf. Comput.}, 185(1):41--55, 2003.

\bibitem{pevzner2000computational}
Pavel Pevzner.
\newblock {\em Computational {M}olecular {B}iology: {A}n {A}lgorithmic {A}pproach}.
\newblock MIT press, 2000.

\bibitem{pe1998median}
Itsik Pe’er and Ron Shamir.
\newblock The median problems for breakpoints are {NP}-complete.
\newblock In {\em Elec. Colloq. on Comput. Complexity}, volume~71, 1998.

\bibitem{popov2007multiple}
V~Yu Popov.
\newblock Multiple genome rearrangement by swaps and by element duplications.
\newblock {\em Theor. Comput. Sci.}, 385(1-3):115--126, 2007.

\bibitem{radcliffe2005reversals}
Andrew~J Radcliffe, Alex~D Scott, and Elizabeth~L Wilmer.
\newblock Reversals and transpositions over finite alphabets.
\newblock {\em SIAM J. Discrete Math.}, 19(1):224--244, 2005.

\bibitem{watterson1982chromosome}
Geoffrey~A Watterson, Warren~J Ewens, Thomas~Eric Hall, and Alexander Morgan.
\newblock The chromosome inversion problem.
\newblock {\em J. Theor. Biol.}, 99(1):1--7, 1982.

\end{thebibliography}

\newpage

\appendix

\section{Preliminaries on parameterized complexity}
\label{sec:preliminariesParameterized}

A \emph{parameterized} problem is a decision problem whose instances are pairs $(x,k) \in \Sigma^* \times \mathbb{N}$, where $k$ is called the \emph{parameter}.
A parameterized problem is \emph{fixed-parameter tractable} (\FPT) if there exists an algorithm $\mathcal A$, a computable function $f$, and a constant $c$ such that given an instance $I=(x,k)$, $\mathcal A$ (called an {\FPT} \emph{algorithm}) correctly decides whether $I \in L$ in time bounded by $f(k) \cdot |I|^c$. 

A parameterized problem is \emph{slice-wise polynomial} ({\XP}) if there exists an algorithm $\mathcal A$ and two computable functions $f,g$ such that given an instance $I=(x,k)$, $\mathcal A$ (called an {\XP} \emph{algorithm}) correctly decides whether $I \in L$ in time bounded by $f(k) \cdot |I|^{g(k)}$. Within parameterized problems, the class {W}[1] may be seen as the parameterized equivalent to the class \NP of classical optimization problems. Without entering into details (see~\cite{downey2012parameterized,cygan2015parameterized} for the formal definitions), a parameterized problem being {W}[1]-\emph{hard} can be seen as a strong evidence that this problem is {not} \FPT. The canonical example of {W}[1]-hard problem is \textsc{Clique} parameterized by the size of the solution.

To transfer {W}[1]-hardness from one problem to another, one uses a \emph{pa\-ra\-me\-te\-ri\-zed reduction}, which given an input $I=(x,k)$ of the source problem, computes in time $f(k) \cdot |I|^c$, for some computable function $f$ and a constant~$c$, an equivalent instance $I'=(x',k')$ of the target problem, such that $k'$ is bounded by a function depending only on $k$. Equivalently, a  problem is W[1]-hard if there is a parameterized 
reduction from \textsc{Clique} parameterized by the size of the solution.


In addition to the {W}[1] class, some classes of pa\-ra\-me\-te\-ri\-zed problems are defined according to their parameterized intractability level. These classes are organized in the so-called W-hierarchy, \FPT $\subseteq$ W[1] $\subseteq$ W[2] $\subseteq$ $\ldots$ $\subseteq$ W[P] $\subseteq$ XP, and it is conjectured that each of the containments is proper~\cite{downey2012parameterized}. If P = \NP, then the hierarchy collapses~\cite{downey2012parameterized}. A parameterized problem is \textit{para-\NP-hard} if it is \NP-hard for some fixed value of the parameter, such as the $k$-{\sc Coloring} problem parameterized by the number of colors for every fixed $k \geq 3$.


\begin{definition}[Bodlaender et al.~\cite{bodlaender2011kernel}]
Let $P, Q \subseteq \Sigma^* \times \mathbb{N}$ be parameterized problems. We say that a polynomial computable function $f: \Sigma^* \times \mathbb{N} \rightarrow \Sigma^* \times \mathbb{N}$ is a \emph{polynomial parameter transformation} (PPT) from $P$ to $Q$ if for all $(x,k) \in \Sigma^* \times \mathbb{N}$ the following holds: $(x,k) \in P$ if and only if $(x’, k’)  = f(x,k) \in Q$ and $k’ \leq k^{O(1)}$.
\end{definition}

\begin{definition}[Bodlaender et al.~\cite{bodlaender2011kernel}]
A \emph{kernelization algorithm}, or in short, a \emph{kernel} for a parameterized problem $L \subseteq \Sigma^* \times \mathbb{N}$ is an algorithm that given $(x,k) \in \Sigma^* \times \mathbb{N}$, outputs in $p(|x| + k)$ time a pair $(x’,k’) \in \Sigma^* \times \mathbb{N}$ such that
\begin{itemize}
\item[$\bullet$] $(x, k) \in L \Leftrightarrow (x’, k’) \in L$, and
\item[$\bullet$] $|x’|, k’ \leq f(k)$,
\end{itemize}
where $f$ is some computable function and $p$ is a polynomial. Any function $f$ as above is referred to as the \emph{size} of the kernel.
\end{definition}

If we have a kernel for $L$, then for any $(x, k) \in \Sigma \times \mathbb{N}$, we can obtain in polynomial time an equivalent instance with respect to $L$ whose size is bounded by a function of the parameter. 
Of particular interest are \emph{polynomial kernels}, which are kernels whose size is bounded by a polynomial function.

\begin{theorem}[Bodlaender et al.~\cite{bodlaender2011kernel}]
Let $P$ and $Q$ be parameterized problems and $P’$ and $Q’$ be, respectively, the unparameterized versions of $P$ and $Q$. Suppose that $P’$ is \NP-hard and $Q’$ is in \NP. Assume that there is a polynomial parameter transformation from $P$ to $Q$. Then if $Q$ admits a polynomial kernel, so does $P$. Equivalently, if $P$ admits no polynomial kernel under some assumption then neither does $Q$.
\end{theorem}

\section{Sorting by rearrangement operations}\label{sec:Operacoes}

\paragraph*{The breakpoint distance.}
An \emph{adjacency in a permutation} 
$\pi$ with respect to permutation $\sigma$ is a pair $(a,e)$ of consecutive elements in $\pi$ such that $(a,e)$ 
is also consecutive in $\sigma$. 
If a pair of consecutive elements is not an adjacency, then $(a,\ e)$ is called a \emph{breakpoint}, and we denote by $d_{\sf BP}(\pi, \sigma)$ the number of breakpoints of~$\pi$ with respect to $\sigma$. 
The set $\adj(\pi)$ of adjacencies of $\pi$ is thus given by $\adj(\pi) = \{\{\pi_i, \pi_{i+1}\} \mid i=1, \cdots, n-1 \}$. 
Thus, in other words, the breakpoint distance between $\pi$ and $\sigma$ is $d_{\sf BP}(\pi, \sigma) = |\adj(\pi) - \adj(\sigma)|$. 

\paragraph*{The block-interchange and the transposition distances.}
Bafna and Pevzner~\cite{Bafna1998Sorting} proposed a useful graph, called the reality and desire diagram, which allowed non-trivial bounds on the transposition distance~\cite{Bafna1998Sorting}, and also provided, as established by Christie~\cite{christie1998genome}, the exact block-interchange distance. 
Nevertheless, when considering the transposition distance, the reality and desire diagram is a tool to only deal with lower and upper bounds for a permutation, as discussed below.

Given a permutation $\pi$ of length $n$, the \emph{reality and desire diagram} $G(\pi, \iota)$ 
(or just $G(\pi)$ when convenient) from $\pi$ to $\iota$, is a multigraph $G(\pi)=(V,R\cup D)$, where $V = \{0, -1, +1, -2, +2, \cdots ,\allowbreak -n, +n, -(n+1)\}$, each element of $\pi$ corresponds to two vertices and we also include the vertices labeled by $0$ and $-(n+1)$, and the edges are partitioned into two sets: the reality edges $R$ and the desire edges $D$. The \emph{reality edges} represent the adjacency between the elements on $\pi$, that is $R = \{(+\pi(i),\ -\pi(i+1))\mid  i=1,\cdots, n-1\} \cup \{(0,\ -\pi(1)),\ (+\pi(n),\ -(n+1))\}$; and the \emph{desire edges} represent the adjacency between the elements on $\iota$, that is $D = \{(+i,\ -(i+1)) \mid i=0,\cdots, n\}$. \autoref{fig:bg1} illustrates the reality and desire diagram of a permutation. A general definition considers $G(\pi,\sigma)$ where the reality (resp. desire) edges represent the adjacency between elements of $\pi$ (resp. $\sigma$), and then $D = \{(+\sigma(i),\ -\sigma(i+1)) \mid i=0,\cdots, n\}$. 

\begin{figure}[h!]
	\centering
			\includegraphics[scale=.3]{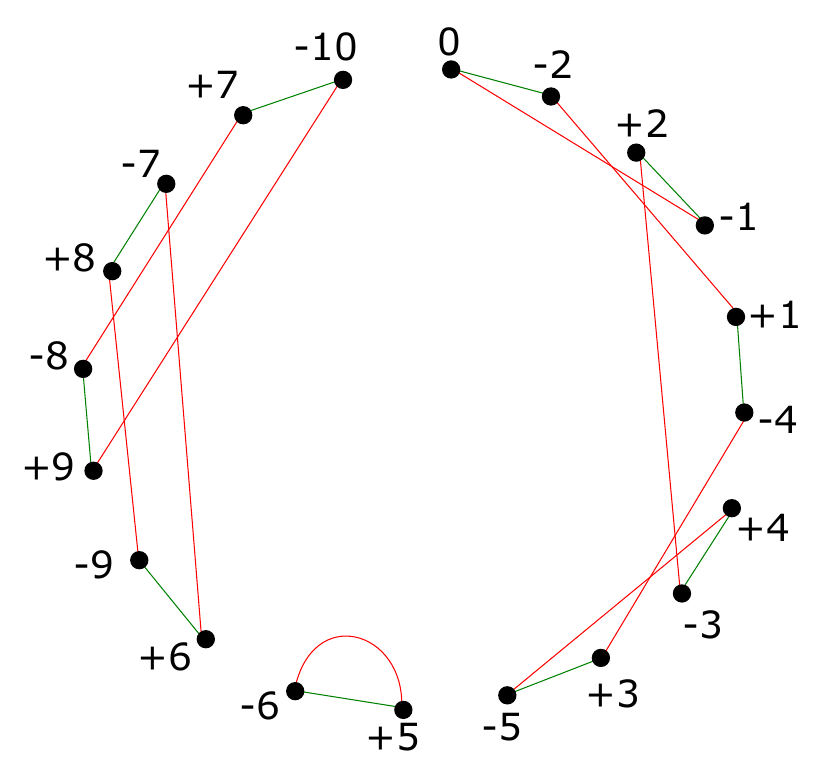}
	\caption{The reality and desire diagram between permutation $[\bm{0}\, 2\, 1\, 4\, 3\, 5\, 6\, 9\, 8\, 7\, \bm{10}]$ and $\iota$, where bold edges are reality edges, and the other ones are the desire edges.}
	\label{fig:bg1}
\end{figure}

As a direct consequence of the construction of this graph, the sets of reality edges and desire edges define two perfect matchings (that is, a set of edges that contains all vertices of the graph and each of them appears exactly once), denoted by $M(\pi)$ and $M(\iota)$, respectively. 
Each of these perfect matchings is called a \emph{permutation matching}. 

Since every vertex in $G(\pi)$ has degree two, $G(\pi)$ can be partitioned into disjoint cycles. We say that a cycle in $\pi$ has length $k$, or that it is a $k$-cycle, if it has exactly $k$ reality edges (or, equivalently, $k$ desire edges). Hence, the identity permutation of length $n$ has $n+1$ cycles of length one. We denote by $c(G(\pi, \iota))$ (or just $c(G(\pi))$ for convenience) the number of cycles in $G(\pi)$. 
One may note that when it is needed to compute properly the distance between two permutations $\pi$ and $\sigma$, we  explicitly denote the corresponding parameters by $G(\pi,\sigma)$ and $c(G(\pi, \sigma))$.
After applying a block-interchange $b\ell$ in a permutation $\pi$, the number of cycles $c(G(\pi))$ changes in such a way that $c(G(\pi b\ell)) = c(G(\pi)) + x$, for some $x \in \{-2, 0, 2\}$ (see~\cite{christie1998genome}). 
The block-interchange $b\ell$ is thus classified as an \emph{$x$-move} for~$\pi$. 
Analogously, after applying a transposition $t$ in a permutation~$\pi$, the number of odd cycles $c_{\sf odd}(G(\pi))$ changes by an \emph{$x$-move}, $x \in \{-2, 0, 2\}$, for~$\pi$~(see~\cite{Bafna1998Sorting}). 
Christie~\cite{christie1998genome} proved, for the block-interchange operation, the existence of a $2$-move for any permutation, which says that the number of cycles yields the exact block-interchange distance:

\begin{theorem}[Christie~\cite{christie1998genome}]\label{thm:BIdistance}
\label{BIdistance}
The block-interchange distance of a permutation $\pi$ of length~$n$ is
$
d_{\sf BI}(\pi) = \frac{(n+1)-c(G(\pi))}{2}.
$
\end{theorem}

On the other hand, by allowing only the particular case of the transposition operation, a $2$-move is not always possible to be used. We say that a transposition \emph{affects} a cycle if the extremities of the two blocks of the transposition eliminate a reality edge of a cycle and creates another edge. This new edge may increase, decrease, or keep the number of cycles.

A transposition $t(i,j,k)$, where $1 \leq i < j < k \leq n+1$, is a permutation that exchanges the contiguous blocks $i \; i\!+\!1 \cdots j\!-\!1$ and $j \; j\!+\!1 \cdots k\!-\!1$; when composed with a permutation $\pi_{[n]}$, it yields the following permutation:
$$\pi_{[n]} \cdot t(i,j,k) = [\pi_1 \ \pi_2 \ \cdots \pi_{i-1} \ \underline{\pi_j \ \cdots \ \pi_{k-1}} \ \underline{\pi_i \ \ldots \pi_{j-1}} \ \pi_k \cdots \pi_n].$$
Bafna and Pevzner~\cite{Bafna1998Sorting} showed conditions of a cycle for a transposition to be an $x$-move. If a transposition $t$ is a $-2$-move, then $t$ affects three distinct cycles. However, if a transposition $t$ is a $0$-move or a $2$-move, then~$t$ affects at least two elements of the same cycle~\cite{Bafna1998Sorting}. 

\begin{theorem}[Bafna and Pevzner~\cite{Bafna1998Sorting}]
\label{Tdistance}
The transposition distance of a permutation $\pi$ of length~$n$ is
$
d_{\sf T}(\pi) \geq \frac{(n+1)-c_{\sf odd}(G(\pi))}{2}.
$
\end{theorem}

Permutations whose transposition distances are equal to the lower bound of \autoref{Tdistance} are called \emph{hurdle-free permutations}~\cite{christie1998genome,bader2011transposition}. 
Cunha et al.~\cite{cunha2014faster,cunha2015faster} showed upper bounds on the distance of any permutation by using permutation trees data structure, and based on that, an $1.375$-approximation algorithm was developed, improving the time complexity to $O(n\log n)$. 
An interesting transformation of permutations is the reduction, since the permutation obtained after its reduction preserves the block-interchange distance and the transposition distance as well. 
The \emph{reduced permutation} of $\pi$, denoted by $gl(\pi)$, is the permutation whose reality and desire diagram $G(gl(\pi))$ is equal to $G(\pi)$ without the cycles of length one, and has its vertices relabeled accordingly. For instance, the reduced permutation corresponding to the permutation in \autoref{fig:bg1} is $[\bm{0}\ 2\ 1\ 4\ 3\ 5\ 8\ 7\ 6\ \bm{9}]$. 

\begin{theorem}[Christie~\cite{christie1998genome}]
\label{BIequalityreduced}
The block-interchange distances of a permutation $\pi$ and its reduced permutation $gl(\pi)$ satisfy
$
d_{\sf BI}(\pi) = d_{\sf BI}(gl(\pi)).
$
\end{theorem}

\paragraph*{Swap distance.} Permutations can also be represented by each element followed by its image. For example, given a set $\{1, 2, 3\}$, the sequence $(1\ 2\ 3)$ maps $1$ into $2$, $2$ into $3$, and $3$ into $1$, corresponding to the permutation $[2 \ 3\ 1]$.
This representation is not unique; for instance, $(2\ 3\ 1)$ and $(3\ 1\ 2)$ are equivalent.
Permutations are composed of one or more algebraic cycles, where each algebraic cycle of a permutation $\pi$ is a representation of a domain $i$, followed by its image $\pi(i)$, followed by getting the image of $\pi(i)$ as the next element, i.e., $\pi(\pi(i))$, and so on. 
We continue this process until we reach a repeated element. 
This procedure uniquely defines the permutation. 
We denote by $c(\pi)$ the number of algebraic cycles of $\pi$.
For example,
given $\pi = [8 \ 5\ 1\ 3\ 2\ 7\ 6\ 4] = (1\ 8\ 4\ 3)(2\ 5)(6\ 7)$, we have $c(\pi) =3$. 
An exchange of elements involving elements $a$ and $b$ such that $a$ and $b$ are in the same cycle is an exchange that \emph{breaks} the cycle in two, whereas if $a$ and $b$ belong to different cycles, the exchange of these elements \emph{unites} the two cycles~\cite{fertin2009combinatorics}.
Thus, when considering the \emph{swap} metric, the \emph{swap} distance of a $\pi$ permutation is determined as follows: $d_{\sf swap}(\pi) = n - c(\pi)$, where $c(\pi)$ is the number of algebraic cycles of $\pi$.

\paragraph*{Short-block-move distance.} A \emph{$p$-bounded block-move} is a transposition $t(i,j,k)$ such that $k - i \leq p$, and a $3$-bounded block-move is called a \emph{short-block-move}. Hence, a short-block-move is either a transposition $t(i,i+1,i+2)$, called a \emph{skip}, a transposition $t(i,i+1,i+3)$, or a transposition $t(i,i+2,i+3)$, the two latter ones called \emph{hops}. If the transpositions are restricted only to $p$-bounded block-moves, then one obtains the \emph{$p$-bounded block-move distance} $d_{\sf pbbm}(\pi_{[n]})$. When $p=3$, this defines the \emph{short-block-move distance} $d_{\sf sbm}(\pi_{[n]})$.
Previous works investigated variants of block-move distances where bounds are imposed on the lengths of at least one of the blocks moved~\cite{Vergara2,Vergara}. The problem of sorting permutations using $2$-bounded block-moves, i.e., adjacent swaps, is easily solved by the \emph{Bubble-Sort} 
algorithm~\cite{book:Knuth}. 
In general, the complexity of the problem of sorting a permutation by $p$-bounded block-moves is unknown for fixed $p > 2$, whereas the analogous problem of limiting $k-i \leq f(n)$, is \NP-hard~\cite{Vergara2}, since {\sc Sorting by Transpositions} is \NP-hard.

To estimate the short-block-move distance, Heath and Vergara \cite{Vergara2,Vergara} used the \emph{permutation graph} $PG(\pi_{[n]}) = (V^p_\pi, E^p_\pi)$, where $V^p_\pi = \{1,  2,  \ldots, n\}$ and $E^p_\pi = \{ (i,j) \mid   \pi_i > \pi_j, \ i < j\}$; each edge of $PG$ is called an \emph{inversion} in $\pi$. Heath and Vergara proved that on a shortest sequence of operations for $\pi_{[n]}$, every short-block-move decreases the number of inversions by at least one unit, and by at most two units, therefore: 
$
 \left\lceil\frac{|E^p_\pi|}{2}\right\rceil \leq d_{\sf sbm}(\pi_{[n]})  \leq |E^p_\pi|.
 $
Given a permutation, our aim is to minimize the number of operations that decrease only one inversion in $PG$. Examples of permutations that are tight with respect to the above  lower and upper bounds are $[2\,\,4\,\,3\,\,5\,\,1]$ and $[2\,\, 1\,\, 4\,\, 3\,\, 6\,\,5]$, respectively. 

A short-block-move is a \emph{correcting} move if it is a skip that eliminates one inversion, or a hop that eliminates two inversions in $\pi$. Otherwise, the block-move is called \emph{non-correcting}. Heath and Vergara~\cite{Vergara} proved that each sorting sequence can be performed by using just correcting moves. \autoref{tab:replacing-beta-is} shows replacements from non-correcting moves to correcting moves in an optimal sorting sequence, which we will use later in \autoref{thm:sorting-each-cc-separately-is-optimal}.

\begin{table}[htbp]
\centering
\begin{footnotesize}
\begin{tabular}{|c||c|c||c|}
\hline
 case & $\pi$ & $\pi'=\pi\beta_i$ & $\pi''=\pi\beta'_i$\\
\hline
1 & $\cdots ef \cdots$  & $\cdots fe \cdots$  & $\cdots ef \cdots$  \\
2 & $\cdots exf \cdots$ & $\cdots xfe \cdots$ & $\cdots xef \cdots$ \\
3 & $\cdots exf \cdots$ & $\cdots fex \cdots$ & $\cdots efx \cdots$ \\
4 & $\cdots xef \cdots$ & $\cdots fxe \cdots$ & $\cdots exf \cdots$ \\
5 & $\cdots efx \cdots$ & $\cdots fxe \cdots$ & $\cdots exf \cdots$ \\
\hline
\end{tabular}
\end{footnotesize}
\medskip
\caption{How to replace a non-correcting move $\beta_i$ with a correcting move $\beta'_i$~\cite{Vergara}; in all cases, $e<f$, and $x$ is arbitrary.}
\label{tab:replacing-beta-is}
\end{table}

\section{Proofs deferred from \autoref{sec:median}}
\label{ap:median-problem}

In this appendix we provide the proofs deferred from \autoref{sec:median}. For the sake of readability, we restate the corresponding result before proving it. 

\bigskip

\noindent
\autoref{thm:bimpNPC}.
\emph{
The {\sc Block-interchange Median} problem is \NP-complete even when the input consists of three permutations.
}
\begin{proof}
Let $\pi^1, \pi^2, \pi^3$ be permutations of $\{1, \cdots ,n\}$, and let $k$ be an integer. 
Then, $(\pi^1, \pi^2, \pi^3, \allowbreak k) \in \cmp$ if and only if there is a permutation $\sigma$ satisfying $\sum_{i=1}^{3} c(G(\sigma,\pi^i)) \geq k$. 
Since solving an $\cmp$ instance is equivalent to finding a permutation matching $M(\sigma)$ such that  $\sum_{i=1}^{3} c(G(\sigma,\pi^i))$ is maximized and the block-interchange distance between any two permutations $d_{\sf BI}(\sigma,\pi^i) = \frac{n+1 - c(G(\sigma,\pi^i))}{2}$, then $\cmp \leq_p \bimp$. 
~\qed\end{proof}

\bigskip

\noindent
\autoref{lm:reducedBI-t}.
\emph{
If an adjacency occurs in all of the input permutations, then it occurs also in solutions of the {\sc Block-interchange Median} problem and the {\sc Transposition Median} problem. 
}
\begin{proof}
Assume ${\tt ab}$ is a motif pair that occurs in all $k$ input permutations. 
Let $\sigma$ be a solution median permutation satisfying $a$ and $b$ are not adjacent. 
Suppose, w.l.o.g. $\sigma = [\sigma_1, \cdots, \sigma_{i-1}, {\tt a}, \sigma_{i+1}, \cdots, \sigma_{j-1}, {\tt b}, \sigma_{j+1}, \cdots, \sigma_n]$. 
Thus, we obtain the permutation $\sigma’$ from $\sigma$ 
by setting adjacencies ${\tt ab}$, ${\tt b}\sigma_{i+1}$ and $\sigma_{j-1}\sigma_{j+1}$,  
removing ${\tt a} \sigma_{i+1}$, $\sigma_{j-1}{\tt b}$ and ${\tt b}\sigma_{j+1}$, 
and keeping all other adjacencies of $\sigma$, 
i.e., $\sigma’ = [\sigma_1, \cdots, \sigma_{i-1}, {\tt ab}, \sigma_{i+1}, \cdots, \sigma_{j-1}, \sigma_{j+1}, \cdots, \sigma_n]$. 
Now, considering any optimum sequence of block-interchanges (or transpositions) from $\sigma$ to $\pi^i$, we are present a simulation sequence from $\sigma'$ to $\pi^i$. 
Any operation applied on a sequence from $\sigma$ to $\pi^i$ that does not change adjacencies ${\tt a} \sigma_{i+1}$, $\sigma_{j-1}{\tt b}$ and ${\tt b}\sigma_{j+1}$ can be simulated properly from $\sigma’$ to $\pi^i$ without any loss, once any impact on the decreasing number of breakpoints is the same, and so the number of cycles on the reality and desire diagram. 
If an operation applied on a sequence from $\sigma$ to $\pi^i$ affects i) ${\tt a} \sigma_{i+1}$, ii) $\sigma_{j-1}{\tt b}$ or iii) ${\tt b}\sigma_{j+1}$, then we simulate it on a sequence from $\sigma$ to $\pi^i$ as follows: i) instead of cut a block just after ${\tt a}$, it is cut after the two elements ${\tt ab}$; ii) instead of cut a block just before ${\tt b}$, it is cut just before $\sigma_{j+1}$; iii) instead of cut a block just after ${\tt b}$ it is cut just before $\sigma_{j+1}$ as well. 
Since all input permutations have the adjacency ${\tt ab}$, no extra operation must be applied from $\sigma'$ to $\pi^i$. 
Thus, we conclude that $\sum\limits_{i=1}^{k} d_{\sf BI}(\sigma’, \pi^i) \leq \sum\limits_{i=1}^{k} d_{\sf BI}(\sigma, \pi^i)$ (or $\sum\limits_{i=1}^{k} d_{\sf T}(\sigma’, \pi^i) \leq \sum\limits_{i=1}^{k} d_{\sf T}(\sigma, \pi^i)$). 
~\qed\end{proof}

\bigskip

\noindent
\autoref{lm:breakpointAdjsolution}.
\emph{
If an adjacency occurs in all of the input permutations, then it occurs also in a solution of the {\sc Breakpoint Median} problem.
}
\begin{proof}
Suppose that $X = [x_1, x_2, \cdots, x_n]$ is a breakpoint median for the input $\pi^1, \pi^2, \cdots, \pi^k$, and $\{x_i, x_j\}$ is a pair in $(\adj(\pi^1) \cap \cdots \cap \adj(\pi^k)) \setminus \adj(X)$. 
We obtain a set $Y$ being $Y = \adj(X) \cup \{\{x_i, x_j\}\}$. 
Hence, we modify $Y$ in such a way to generate a set of adjacencies which forms a median solution. 
Given a pair of adjacency $\{x,y\} \in X$, let $w(x,y) = |\{X \in \{\pi^1, \cdots, \pi^k\} \ : \ \{x,y\} \in \adj(X)\}|$, i.e., $w(x,y)$ is the number of input permutations that have the adjacency $\{x,y\}$. 
If $w(x_{i-1}, x_i) \leq w(x_i, x_{i+1})$ we remove $\{x_{i-1}, x_i\}$ from $Y$, otherwise we remove $\{x_i, x_{i+1}\}$. 
In the same way, 
if $w(x_{j-1}, x_j) \leq w(x_j, x_{j+1})$ we remove $\{x_{j-1}, x_j\}$ from $Y$, otherwise we remove $\{x_j, x_{j+1}\}$. 
Note that in each one of the four possible cases, in the resulting set $Y$,  $\{x_i, x_j\}$ happens exactly once, the same as the others elements. 
Since $w(x_i,x_j)=k$, $\sum\limits_{i=1}^{k} d_{\sf BP}(\pi^i, Y) < \sum\limits_{i=1}^{k} d_{\sf BP}(\pi^i, X)$, which is a contradiction. 
~\qed\end{proof}

\bigskip

\noindent
\autoref{thm:kernel}.
\emph{
The following problems admit a polynomial kernel parameterized by the value $d$ of the desired  median solution:}
\emph{
\begin{enumerate}
    \item {\sc $d$-Swap Median}.
    \item {\sc $d$-Breakpoint Median}.
    \item {\sc $d$-Block-interchange Median}.
    \item {\sc $d$-Transposition Median}.
    \item {\sc $d$-Short-Block-Move Median}.
\end{enumerate}
}
\begin{proof}
First, we consider a polynomial kernelization for the {\sc $d$-Swap Median} problem based on the following reduction rules:

\begin{enumerate}
    \item If there is a column with more than $d+1$ elements, then return {\tt no}.
    \item If there is a column with at least two elements occurring at least $d+1$ times, then return {\tt no}.
    \item If an element occurs more than $d$ times in at least two columns, then return {\tt no}. 
    \item If a row has at least $d$ copies in the matrix, then either the solution is a copy of such a row, or the answer is {\tt no}.

    \smallskip
    
    \item[] We say that a column $i$ is \textit{heavy} for an element $x$ if $x$ occurs more than $d$ times in it; otherwise, $i$ is said \textit{light} for $x$.   

 \smallskip 
 
    \item For each element $x$, if the sum of occurrences of $x$ in its light columns is more than $2d$, then return {\tt no}.
    \item If the previous rules were not applied, remove the columns whose all elements are the same, and reduce the universe accordingly.
\end{enumerate}

Since the goal is to determine whether there is a permutation $s^{\star}$ whose sum of the distances by swaps between $s^{\star}$ and all permutations of $S$ is at most $d$, Rules 1-4 are clearly safe. 
Now, we discuss Rule 5. If $S$ is a {\tt yes}-instance then an element $x$ having a heavy column $i$ (by Rule 2, there is at most one heavy column) must have $x$ in the position $i$ of any optimum solution $s^{\star}$. Thus, the number of rows that contains $x$ in positions different than $i$ is at most $d$. Also, if $x$ has no heavy column, then such a $s^{\star}$ contains $x$ in some position $i$ whose column has at most $d$ occurrences of $x$, while the number of rows having $x$ in other positions is also at most $d$ (hence, $2d$ occurrences in total). Thus, Rule 5 is safe.

Regarding Rule 6, as each $s\in S$ is a permutation, it holds that if a column $i$ of $S$ contains only one element $x$, then all permutations of $S$ have $x$ in position $i$, implying that any optimal solution for the problem should contain $x$ in position $i$. Thus, it is safe to ignore that column $i$ and element $x$ from the input. (Recall that $s^{\star}$ having $x$ in position $i$ implies that for any $s\in S$, there is an algebraic cycle of length one between $s^{\star}$ and $s$, which is the best possible because the swap metric can be seen as the minimum number of swaps to get only algebraic cycles of size one.)  

At this point, we may suppose that Rules 1-5 were not applied and that $S'$ is the resulting instance after the application of Rule 6.

To complete the kernelization algorithm for {\sc $d$-Swap Median}, we need the following lemma. 


\begin{lemma}
If $S'$ is a {\tt yes}-instance of the {\sc $d$-Swap Median} problem, then $S'$ has at most $2d$ columns and $4d^2+d$ rows.
\end{lemma}
\begin{proof}
For a column $i$ having more than one element, the distance between the optimum solution $s^{\star}$ and some permutation of $S$ needs to count a swap move involving column $i$. In addition, each swap affects only two columns, implying that $d$ moves can affect at most $2d$ columns. However, by Rule 6, any column of $S'$ has more than one element. Therefore, if $S'$ is a {\tt yes}-instance of the {\sc $d$-Swap Median} problem, it must have at most $2d$ columns.    

Now, let us argue about the number of rows of $S'$. 
Recall that a row $s$ is light if it has an element $x$ in position $i$ such that column $i$ is light for $x$; otherwise it is heavy. By Rule 5, the number of rows that contain an element $x$ in a light column for it is at most $2d$. Since, $S'$ has at most $2d$ columns it also has at most $2d$ elements. Therefore, the number of light rows is at most $4d^2$. Finally, by definition, heavy rows have only elements in positions for which they are heavy. By Rule 3, each element is heavy for only one column, which implies that heavy columns are copies. By Rule 4, we conclude that we have at most $d$ heavy rows in $S'$. Hence, $S'$ has at most $4d^2+d$ rows. 
~\qed\end{proof}

Therefore, either the size of $S'$ certifies a {\tt no}-answer, or $S'$ is returned as a kernel for the {\sc $d$-Swap Median} problem.  

 \bigskip

Next, we discuss a kernelization algorithm for the {\sc $d$-Breakpoint Median} problem.
Recall that in the breakpoint metric, one does not care about occurrences of elements in columns but adjacencies of elements instead (regardless of their position in the rows). Thus, we should adapt the previous arguments accordingly. 

A kernel for the {\sc $d$-Breakpoint Median} problem can be found as follows:

\begin{enumerate}

    \item If there is an element having with more than $d+1$ distinct successor elements (distinct adjacencies) in the matrix, then return {\tt no}.

    \item If there is an element $x$ with at least two elements occurring at least $d+1$ times as successor of $x$ in the matrix, then return {\tt no}.

    \item If an element occurs more than $d$ times as successor of at least two other elements in the matrix, then return {\tt no}. 
    
    \item If a row has at least $d$ copies in the matrix, then either the solution is a copy of such a row, or the answer is {\tt no}.

    \smallskip
    
    \item[] We say that an element $y$ is a \textit{heavy successor} for an element $x$ if $xy$ occurs more than $d$ times in the matrix; a successor of $x$ that is not heavy is said to be a \textit{light successor} for $x$.   

    \smallskip

    \item For each element $x$, if the sum of occurrences of $x$ with light successors is more than $2d$, then return {\tt no}.
    
    \item Assuming that the previous rules were not applied, if there is an adjacency between $x$ and $y$ (i.e., $xy$) occurring in all of the input permutations, then consider $xy$ as a single element and reduce the universe  accordingly. Repeat this until there is no such adjacencies.
\end{enumerate}

The safety of Rules 1-4 is straightforward, for Rule~5 the argument is similar to the swap case replacing columns by successors, and for Rule~6 the safety proof follows from \autoref{lm:breakpointAdjsolution}. Again, we suppose that Rules 1-5 were not applied and $S'$ is the resulting instance after the application of Rule 6.

Similarly as above, to complete the kernelization algorithm for {\sc $d$-Breakpoint Median}, we need the following lemma.

\begin{lemma}~\label{lm:kernel_BK}
If $S'$ is a {\tt yes}-instance of the {\sc $d$-Breakpoint Median} problem, then $S'$ has at most $2d$ columns and $4d^2+d$ rows.
\end{lemma}
\begin{proof}
For an element $x$ having more than one element as successor or more than one element as predecessor in the matrix, the distance between the optimum solution $s^{\star}$ and some permutation of $S$ needs to count a breakpoint involving element $x$. In addition, each breakpoint involves only two elements, implying that $d$ breakpoints can involve at most $2d$ elements. However, by Rule 6, any element of $S'$ is involved in at least one breakpoint. Therefore, if $S'$ is a {\tt yes}-instance of the {\sc $d$-Breakpoint Median} problem, then, since we are dealing with permutations,  it must have at most $2d$ elements and at most $2d$ columns.

Now, let us argue about the number of rows of $S'$.  Recall that a row $s$ is light if it has an adjacency $xy$ such that $y$ is a light successor for $x$; otherwise it is heavy. By Rule 5, for each element $x$, the number of rows that contain an adjacency $xy$ where $y$ is a light successor of $x$ is at most $2d$. Since, $S'$ has at most $2d$ elements, the number of light rows is at most $4d^2$. Finally, by definition, heavy rows have only heavy successors. By Rule 3, each element is heavy for only one predecessor in the matrix, which implies that heavy rows are copies. By Rule 4, we conclude that we have at most $d$ heavy rows in $S'$. Hence, $S'$ has at most $4d^2+d$ rows.
~\qed\end{proof}


Therefore, as above, either the size of $S'$ certifies a {\tt no}-answer, or $S'$ is returned as a kernel for the {\sc $d$-Breakpoint Median} problem. 

\bigskip

Next, we discuss a kernelization for the {\sc $d$-Block-interchange Median} problem and the {\sc $d$-Transposition Median} problem. Recall that for both metrics, whenever there is a breakpoint there is a move to be ``played'' to obtain the identity. Thus, a large set of breakpoints being one per row is enough to certify a {\tt no}-answer for both problems as well. As Rules 1-5 of the previous kernelization deal only with these kind of sets of breakpoints, they also hold as reduction rules for these two problems. 

On the other hand, an analogous of Rule~6 may depend on the kind of move to be used. However, \autoref{lm:reducedBI-t} shows that a similar reduction rule can also be applied for the {\sc $d$-Block-interchange Median} problem and the {\sc $d$-Transposition Median} problem. Regarding an analogous of \autoref{lm:kernel_BK}, it is enough to observe that any block-interchange involves the adjacency of at most eight elements (at most four adjacencies involved), and then one can conclude that $S'$ has at most $8d$ elements/columns and $16d^2+d$ rows. Similarly, concerning transpositions, each move involves the adjacency of at most six elements (at most three adjacencies involved), and then one can conclude that $S'$ has at most $6d$ elements/columns and $12d^2+d$ rows. 

\bigskip

Finally, we discuss a kernelization for the {\sc $d$-Short-Block-Move Median} problem.
As previously discussed,  Rules 1-5 described for the breakpoint distance can be also applied to any metric where the existence of a breakpoint certifies the existence of a move to be ``played'' in order to obtain the identity. Thus, they work for the short-block-move distance as well.  However, unlike with the {\sc $d$-Block-interchange Median} problem and the {\sc $d$-Transposi\-tion Median} problem, an immediate analogue of Rule 6 does not apply to the short-block-move distance, because it may be necessary to traverse some positions to get an element from one point to another, temporarily breaking some ``good'' adjacencies.
To get around this problem, we introduce the notion of homogeneous columns. 

A column of an input matrix/set $S$ is \textit{homogeneous} if it contains only one element, and \textit{heterogeneous} otherwise. Note that the existence of a heterogeneous column implies the existence of a move involving such a column. Since we are looking for a permutation $s^{\star}$ whose sum of distances from the input permutations is at most $d$, it follows that $S$ contains at most $3d$ heterogeneous columns. So, either we have already a kernel or too many homogeneous columns where many of them are not involved in moves needed for the calculation of the distance between $s^{\star}$ and any $s\in S$. Then, an analogous of Rule 6 for this problem must identify these homogeneous columns, remove them, and reduce the universe properly. Due to \autoref{lm:sbmReduction}, we can safely apply the following reduction rule.

\begin{itemize}
    \item[$\star$] If there is an interval $I$ with $6d+2$ consecutive homogeneous columns, then remove the middle columns of $I$ and reduce the universe size accordingly. Repeat this until there is no such interval.
\end{itemize}

After applying the above rule, we claim that the number of columns of a {\tt yes}-instance is at most $18d^2+9d+1$, because it has at most $3d$ heterogeneous columns and a sequence of at most $6d+1$ homogeneous columns before/after a heterogeneous one. This remark together with the reduction rules applied implies that the number of rows is at most $36d^3+18d^2+3d$. This concludes the existence of a polynomial kernel for the {\sc $d$-Short-Block-Move Median} problem.
~\qed\end{proof}

\section{Proofs deferred from \autoref{sec:closest}}
\label{ap:closest-problems}

In this appendix we provide the proofs deferred from \autoref{sec:closest}. For the sake of readability, we restate the corresponding result before proving it. 

\bigskip

\noindent
\autoref{lm:sumequal}.
\emph{
Given permutations $\pi^1$ and $\pi^2$, we have that $d_{\sf BI}(\pi^{1,2}) = d_{\sf BI}(\pi^1) + d_{\sf BI}(\pi^2)$.
}
\begin{proof}
Assuming that $\pi^1$ has $p$ elements and $\pi^2$ has $q$ elements, 
since $p+1$ is greater than all elements of $\pi^1$ and smaller than all elements of $\pi^2$, the reality and desire diagram $G(\pi^1 \uplus \pi^2)$ is obtained by gluing $G(\pi^1)$ and $G(\pi^2)$, i.e., the reality and the desire edges do not change when the union operation is applied to permutations. 
As direct consequence of \autoref{thm:BIdistance}, we have $d_{\sf BI}(\pi^{1,2}) = \frac{p + q + 2 - c(G(\pi^1)) -c(G(\pi^2)}{2} = d_{\sf BI}(\pi^1) + d_{\sf BI}(\pi^2)$.
~\qed\end{proof}

\bigskip

\noindent
\autoref{thm:closestmedian}.
\emph{
For any three permutations $\pi^1, \pi^2, \pi^3$, we have that $\sigma$ is solution of $\bimp$ if and only if $\biguplus\limits_{i=1}^{6} \sigma$ is solution of $\bicp$ for the permutations 
$\pi^{1,2,3,1,2,3}, \pi^{2,1,1,3,3,2}$, and $\pi^{3,3,2,2,1,1}$.
}
\begin{proof}
Since permutations $\pi^{1,2,3,1,2,3}, \pi^{2,1,1,3,3,2}$, and $\pi^{3,3,2,2,1,1}$ are composed by six parts of unions and, considering $\bicp$, each part corresponds to columns yielding $\pi^1$, $\pi^2$, and~$\pi^3$. 
Moreover, by \autoref{lm:sumequal} each part can be treated separately without loss of optimality. 
Hence, there is a solution of $\bicp$ where all parts have the same solution $\delta$. 
Therefore, there is a permutation $x\in \{\pi^{1,2,3,1,2,3}, \pi^{2,1,1,3,3,2}, \allowbreak\pi^{3,3,2,2,1,1}\}$ such that $d_{\sf BI}(\biguplus\limits_{i=1}^{6}\delta, x) = 2(d_{\sf BI}(\delta, \pi^1)+ d_{\sf BI}(\delta, \pi^2)+ d_{\sf BI}(\delta, \pi^3))$. 
Since we want $\delta$ such that $d_{\sf BI}(\biguplus\limits_{i=1}^{6}\delta, x)$ is minimized, we want $\delta$ such that $d_{\sf BI}(\delta, \pi^1)+ d_{\sf BI}(\delta, \pi^2)+ d_{\sf BI}(\delta, \pi^3)$ is minimized. Hence, this happens if and only if $\delta = \sigma$, where $\sigma$ is the solution of $\bimp$.
~\qed\end{proof}

\bigskip

\noindent
\autoref{cor:TranspClosest3NPc}.
\emph{
{\sc Transposition Closest} is \NP-hard even when the input consists of three permutations which are unions of hurdle-free permutations.
}
\begin{proof}
Bader~\cite{bader2011transposition} proved that {\sc Transposition Median} is \NP-hard when $k=3$ even for hurdle-free permutations, i.e., permutations in which the transposition distances are equal to the lower bound of \autoref{Tdistance}. 
Since the distance of unions of hurdle-free permutations can be obtained by the sum of the distances of each part of the union, \autoref{thm:closestmedian} holds in the same~way.
~\qed\end{proof}

\bigskip

\noindent
\autoref{thm:sorting-each-cc-separately-is-optimal}.
\emph{
 For every permutation $\pi$, sorting each connected component of $\pi$ separately is optimal.
}
\begin{proof}
We allow ourselves to use merging moves, which can be replaced by correcting moves as in \autoref{tab:replacing-beta-is}. The modified sequence is not longer than the original, and we observe that these new moves never merge components.

A merging move must act on contiguous components of $\pi$. Let us assume that the leftmost component the move acts on ends with elements $a$ and $b$, and that the rightmost component starts with elements $c$ and $d$, as represented below:

\begin{center}
 \begin{tikzpicture}
\node[anchor = base] (a) at (-1.5, 0) {$a$};
\node[anchor = base] (b) at (-1, 0) {$b$};
\node[anchor = base] (c) at (0, 0) {$c$};
\node[anchor = base] (d) at (.5, 0) {$d$};

\draw (-2.5, .4) -- (-.6, .4) -- (-.6, -.1) -- (-2.5, -.1);
\draw (1.5, .4) -- (-.4, .4) -- (-.4, -.1) -- (-.4, -.1) -- (1.5,-.1);

\end{tikzpicture}
\end{center}

It implies that $a<c, a<d, b<c$, and $b<d$. 
We now replace any merging move involving those component's extremities with correcting moves. There are five cases to consider:

\begin{itemize}
\item[$\bullet$]  $\underline{a}\ \underline{b\ c}\ d \rightarrow b\ c\ a\ d$: this move satisfies the conditions of case 2 in 
\autoref{tab:replacing-beta-is}, 
so we replace it 
with 
$\underline{a}\ \underline{b}\ c\ d \rightarrow b\ a\ c\ d$.

\item[$\bullet$] $\underline{a\ b}\ \underline{c}\ d \rightarrow c\ a\ b\ d$: this move satisfies the conditions of case 4 in 
\autoref{tab:replacing-beta-is}, 
so we replace it
with 
$\underline{a}\ \underline{b}\ c\ d \rightarrow b\ a\ c\ d$.

\item[$\bullet$] $a\ \underline{b}\ \underline{c}\ d \rightarrow a\ c\ b\ d$: this move satisfies the conditions of case 1 in 
\autoref{tab:replacing-beta-is}, and in this case we just remove that block-move from the sorting sequence.

\item[$\bullet$]  $a\ \underline{b}\ \underline{c\ d} \rightarrow a\ c\ d\ b$: this move satisfies the conditions of case 5 in 
\autoref{tab:replacing-beta-is}, 
so we replace it 
with 
$a\ b\ \underline{c}\ \underline{d} \rightarrow a\ b\ d\ c$.

\item[$\bullet$] $a\ \underline{b\ c}\ \underline{d} \rightarrow a\ d\ b\ c$: this move satisfies the conditions of case 3 in 
\autoref{tab:replacing-beta-is}, 
so we replace it 
with 
$a\ b\ \underline{c}\ \underline{d} \rightarrow a\ b\ d\ c$.
\end{itemize}

\noindent None of the correcting moves that we use to replace the non-correcting moves in those five cases is a merging move, and no such replacement increases the length of our sorting sequence. Given any sorting sequence, we repeatedly apply the above transformation to the merging move with the smallest index until no such move remains; in particular, the transformation applies to optimal sequences as well, and the proof is complete.
~\qed\end{proof}

\bigskip

\noindent
\autoref{lm:iffBI}.
\emph{
Given a set of $k$ permutations obtained by \autoref{alg:sBI}, there is a short-block-move closest permutation with maximum distance at most $d$ if and only if there is a Hamming closest string with maximum distance equal to $d$.
}
\begin{proof}
If $\lambda'$ can be built by \autoref{alg:sBI} for some input string $s'$, then, by \autoref{lm:lengthreduced}, $s'$ is a closest string. 
Otherwise, we search from  left to right along the permutation to find the first position where the corresponding element is different from the one intended to be by the algorithm, which is a position $x \in \{2i-1, 2i\}$. In this case, all elements from position $x$ until the position where the first element $y \in \{2i-1, 2i\}$ appear form inversions with respect to each input permutation, implying the short-block-move distance between the solution $[A\, x\, B\, y\, C]$ and any input greater than the distance between the new permutation $[A\, y\, x\, B\, C]$ and any input permutation, such that $A, B$, and $C$ are blocks of elements. 
By repeating this process, a string agreeing with the algorithm output can be found and, by \autoref{lm:lengthreduced}, a string with maximum distance at most $d$ can be constructed. 
Given a solution string $s$, we obtain the associated permutation $\lambda_s$ given by \autoref{alg:sBI}. By \autoref{lm:lengthreduced} we have the solution $s$ regarding the closest string corresponding to the permutation $\lambda_s$ with the same value of maximum distance~$d$, concluding the proof of the lemma.
~\qed\end{proof}





\end{document}